\pdfoutput=1
\documentclass[journal,10pt]{IEEEtran}

\usepackage[T1]{fontenc}

\usepackage{graphicx}
\usepackage{epstopdf}
\usepackage[export]{adjustbox}
\usepackage[dvipsnames]{xcolor}
\usepackage{bbm}
\usepackage{multirow}

\graphicspath{{./figures/}}
\usepackage{subfigure}
\ifCLASSINFOpdf
\else
\fi
 
\usepackage{amssymb}
\usepackage{amsmath} 
\usepackage[cmintegrals]{newtxmath}
\usepackage{stackengine}
\def\delequal{\mathrel{\ensurestackMath{\stackon[1pt]{=}{\scriptstyle\Delta}}}}

\usepackage{algorithmic}
\usepackage{algorithm}

\usepackage{amsthm}
\usepackage{xpatch}

\newtheorem{lem}{Lemma}

\newtheorem{rem}{Remark}

\xpatchcmd{\proof}{\hskip\labelsep}{\hskip5\labelsep}{}{}

\usepackage{mathtools}
\DeclarePairedDelimiter{\ceil}{\lceil}{\rceil}

\usepackage[normalem]{ulem} 

\usepackage{cite}

\usepackage{array}

\makeatletter
\let\sout@orig\sout
\renewcommand{\sout}[1]{\ifmmode\text{\sout@orig{\ensuremath{#1}}}\else\sout@orig{#1}\fi}
\makeatother

\begin{document}
\title{Reinforcement Learning-Based Secure Beamforming Against Satellite Eavesdroppers}

\author{Juhwan Seo, Hyesang Cho, \textit{Member}, \textit{IEEE}, and Dong-Hyun Jung, \textit{Member}, \textit{IEEE}\\
      \thanks{J. Seo is with Samsung Electronics, Suwon 16677, South Korea (e-mail:juhwanseo92@gmail.com).}
     \thanks{H. Cho is with the Department of Electrical and Electronic Engineering, Inha University, Incheon 22212, South Korea (e-mail: hyesang@inha.ac.kr)}
     \thanks{D.-H. Jung is with the School of Electronic Engineering, Soongsil University, Seoul 06978, South Korea (e-mail: dhjung@ssu.ac.kr).}
}
\maketitle

\begin{abstract}
    This paper investigates physical-layer security for uplink low Earth orbit (LEO) satellite communications in the presence of multiple non-colluding satellite eavesdroppers. Secure beamforming design in such systems is challenging due to time-varying orbital geometry and probabilistic fading-induced outage constraints. To address this, we first derive tractable closed-form expressions for both connection and secrecy outage probabilities under Nakagami-$m$ fading, and develop differentiable upper-bound cost functions that are amenable to optimization. Next, to exploit the predictable orbital dynamics and temporal correlation of satellite mobility, we reformulate the non-convex secrecy rate maximization problem as a constrained Markov decision process. We then develop a primal-dual soft actor-critic algorithm with a multi-head cost critic that jointly optimizes beamforming while enforcing average outage constraints via Lagrangian relaxation. Numerical results show that the proposed framework improves the ergodic secrecy rate over maximum ratio transmission across all eavesdropper configurations, and outperforms zero-forcing in dense eavesdropping regimes. It achieves within 7\% of an offline successive convex approximation benchmark while requiring only a single forward pass, enabling low-complexity real-time operation. These results indicate that the proposed approach is applicable to secure beamforming in dynamic LEO satellite environments.
\end{abstract}


\IEEEpeerreviewmaketitle

\def\E{\mathbb{E}}
\def\P{\mathbb{P}}

\def\re{r_{\mathrm{E}}}

\def\th{\mathrm{th}}
\def\out{\mathrm{out}}
\def\ml{\mathrm{ml}}
\def\sl{\mathrm{sl}}

\def\ser{\mathrm{s}}
\def\eve{\mathrm{e}}
\def\ter{\mathrm{t}}
\def\obj{\mathrm{obj}}
\def\qos{\mathrm{QoS}}
\def\T{\mathrm{T}}
\def\H{\mathrm{H}}

\def\Os{\mathcal{O}_{\ser}}
\def\Oe{\mathcal{O}_{\eve}}
\def\Ok{\mathcal{O}_k}
\def\as{a_{\ser}}
\def\ae{a_{\eve}}
\def\ak{a_{k}}
\def\aols{u_{\ser}}
\def\aolsn{{u_{\ser,n}}}
\def\aole{u_{\eve}}
\def\aolen{{u_{\eve,n}}}
\def\aolk{u_{k}}
\def\aolkn{{u_{k,n}}}
\def\incs{i_{\ser}}
\def\ince{i_{\eve}}
\def\inck{i_{k}}
\def\raans{\Omega_{\ser}}
\def\raane{\Omega_{\eve}}
\def\raank{\Omega_{k}}
\def\Deltas{\Delta_{\ser}}
\def\Deltae{\Delta_{\eve}}
\def\Deltak{\Delta_{k}}
\def\Gamkn{\Gamma_{k,n}}
\def\Gamsn{\Gamma_{\ser,n}}
\def\Gamen{\Gamma_{\eve_j,n}}
\def\gamkn{\Gamma_{k,n}}
\def\gamsn{\Gamma_{\ser,n}}
\def\gamen{\Gamma_{\eve_j,n}}

\def\Mt{M_{\mathrm{t}}}

\def\Ns{N}
\def\Ne{N_{\eve}}

\def\wn{\mathbf{w}_{n}}
\def\Pn{P_n}
\def\hkn{\mathbf{h}_{k,n}}
\def\hsn{\mathbf{h}_{\ser,n}}
\def\hen{\mathbf{h}_{\eve,n}}
\def\akn{\mathbf{a}_{k,n}}
\def\asn{\mathbf{a}_{\ser,n}}
\def\aen{\mathbf{a}_{\eve,n}}
\def\Gkn{G_{k,n}}
\def\Gsn{G_{\ser,n}}
\def\Gen{G_{\eve,n}}
\def\Rsn{R_{\ser,n}}
\def\Ren{R_{\eve,n}}
\def\Rkn{R_{k,n}}
\def\Pcon{P_n^{\mathrm{co}}}
\def\Pson{P_{j,n}^{\mathrm{so}}}

\def\fc{f_{\mathrm{c}}}
\def\Gkmax{G_{k}^{\mathrm{max}}}


\section{Introduction}\label{sec:intro}
\IEEEPARstart{T}{he}
3rd Generation Partnership Project (3GPP) has been investigating the integration between terrestrial networks (TNs) and non-terrestrial networks (NTNs) since Release 15~\cite{TR38.811,TR38.821}. By incorporating the wide coverage of NTN elements, such as geostationary orbit (GEO) and low Earth orbit (LEO) satellites, communication services can be extended far beyond the limitations of terrestrial infrastructure, which could enhance global connectivity. The NTNs can also provide coverage to aerial users such as drones, planes, and urban air mobility vehicles. In the forthcoming 6G standard, 3GPP is expected to make a unified standard for TNs and NTNs. As these integrated networks extend connectivity to diverse users and environments, ensuring the confidentiality of transmissions over satellite links against unauthorized interception becomes an increasingly critical design consideration.

Physical layer security (PLS) exploits the inherent randomness of wireless channels to provide information-theoretic secrecy guarantees, as originally established in \cite{wyner75}. While PLS has been extensively studied in terrestrial networks \cite{mukherjee14}, its application to NTNs has attracted growing attention in recent years~\cite{JSAC_Zhu,TIFS_Lei,TWC_Zheng,lin19_robust_sat,guo20_pls_sat,lin18_cog_sat,WCL_Li}. For instance, the ergodic secrecy capacity in unmanned aerial vehicle (UAV) networks was studied in~\cite{JSAC_Zhu}, where a jamming strategy was proposed to confuse eavesdroppers randomly located on the ground. In
\cite{TIFS_Lei}
and \cite{TWC_Zheng}, zero-forcing (ZF)-based beamforming schemes were developed for multi-beam satellite systems, aiming to minimize the satellite's transmit power while maintaining a secrecy rate constraint. More recently, robust secure beamforming under imperfect channel state information (CSI) was investigated in~\cite{lin19_robust_sat} for multibeam satellite systems, and a threshold-based scheduling scheme for multiuser satellite PLS was proposed in~\cite{guo20_pls_sat}. Joint beamforming designs for cognitive satellite-terrestrial networks were also studied in
\cite{lin18_cog_sat}
and \cite{WCL_Li}. However, these studies focus solely on ground-based eavesdroppers, which intercept the signals from aerial nodes.

Beyond conventional ground-based eavesdroppers, aerial eavesdroppers have recently attracted growing attention, where the secrecy performance against UAV-based eavesdroppers was investigated in~\cite{yuan19_uav_eve, tang19_uav_eve, bao20_uav_eve}. Satellite-based eavesdroppers have also emerged as an increasingly relevant threat due to the rapid proliferation of LEO mega-constellations, which increases the likelihood of unauthorized interception from space by reconnaissance satellites or compromised LEO satellites operating in adjacent orbital planes. In~\cite{my22TIFS}, the secrecy performance of ground-to-satellite uplink transmissions was analyzed when satellites serve as eavesdroppers using a stochastic geometry framework, showing that orbital geometry significantly affects the secrecy capacity.
Unlike UAV eavesdroppers, which are characterized by limited operational range and quasi-static hovering positions, satellite eavesdroppers follow deterministic orbital trajectories determined by Keplerian mechanics with rapidly time-varying channel geometry due to high orbital velocities. These two distinctive properties, namely predictable orbits and rapidly changing channels, pose new security design challenges that are fundamentally different from both ground-based and UAV eavesdropper scenarios.

Reinforcement learning (RL) has been widely adopted for optimizing various aspects of wireless communication systems, such as joint power control and beamforming in terrestrial 5G networks~\cite{23_DRL_5G}, dynamic resource allocation for multibeam satellites~\cite{23_sat_rl_beam}, and energy-efficient beamforming for integrated satellite-aerial-terrestrial networks~\cite{24_sat_rl_beam}. RL can learn adaptive policies without explicit channel models or per-slot iterative optimization.
Designing secure beamforming policies against satellite eavesdroppers involves non-convex optimization problems that must be re-solved at every time slot as the orbital geometry evolves. Conventional iterative algorithms, such as semidefinite relaxation and successive convex approximation (SCA), incur high per-slot computational complexity and require careful per-slot initialization, making them impractical for real-time decision-making in dynamic LEO satellite environments. An RL-based framework is therefore needed to learn beamforming policies that exploit orbital predictability under probabilistic outage constraints with multiple satellite eavesdroppers.

Motivated by this, in this paper, we address this gap by adopting an RL-based approach, where the secure beamforming problem is formulated as a constrained Markov decision process (CMDP) \cite{book99_cmdp, 06_RL_constrained}. We propose a primal-dual soft actor-critic (PD-SAC) algorithm that learns secure beamforming policies under connection and secrecy outage probability constraints in the presence of multiple satellite eavesdroppers. This approach relies on closed-form outage expressions under Nakagami-$m$ fading. Combined with a closed-form inequality for the incomplete gamma function, it yields conservative closed-form cost functions that upper bound the outage probabilities using only elementary operations. These cost functions enable gradient-based RL to directly handle probabilistic outage constraints.
The proposed PD-SAC algorithm employs Lagrangian relaxation techniques to enforce the outage requirements within the CMDP formulation.
The main contributions of this paper are summarized as follows:
\begin{itemize}
    \item To the best of our knowledge, this is the first study to investigate secure uplink beamforming for LEO satellite communication systems in the presence of multiple satellite eavesdroppers. While prior works have mainly considered ground-based eavesdroppers~\cite{JSAC_Zhu,TIFS_Lei,TWC_Zheng,lin19_robust_sat,guo20_pls_sat,lin18_cog_sat,WCL_Li}, they do not capture the distinctive characteristics of satellite eavesdroppers, such as orbital trajectories and time-varying channel geometry. Although satellite eavesdroppers were considered in~\cite{my22TIFS}, secure beamforming design was not addressed. In contrast, this work explicitly incorporates multiple satellite eavesdroppers into the uplink beamforming problem and accounts for their geometry-driven, time-varying channels.
    \item We derive closed-form expressions for the connection and secrecy outage probabilities under Nakagami-$m$ fading. We develop tractable cost functions based on a closed-form gamma function inequality that serve as conservative upper bounds on the outage probabilities. These cost functions enable the use of gradient-based policy optimization for the probabilistic outage constraints.
    \item We reformulate the non-convex secrecy rate maximization problem as a CMDP with average outage constraints. We propose a PD-SAC algorithm with a multi-head cost critic that jointly optimizes the beamforming policy through Lagrangian relaxation. We also analyze the computational complexity of the proposed algorithm against conventional beamformers and the offline SCA benchmark, confirming its suitability for real-time deployment.
    \item We provide simulation results demonstrating that the proposed RL policy outperforms maximum ratio transmission (MRT) and surpasses zero-forcing when the number of eavesdroppers is large, approaching an offline per-slot SCA benchmark while requiring only a single forward pass at inference. The off-policy PD-SAC further outperforms an on-policy primal-dual proximal policy optimization (PD-PPO) counterpart under tight outage constraints, highlighting the role of sample efficiency in constrained RL.
\end{itemize}



\textit{Notations:} The superscript $\T$ indicates the transpose operation.
The absolute value of a complex number $x$ is $|x|$, and the $\ell_2$-norm of a vector $\mathbf{x}$ is $\|\mathbf{x}\|$.
$\mathbf{0}$ and $\mathbf{I}$ denote the all-zero vector and the identity matrix of appropriate dimensions, respectively.
The first-kind Bessel function of order $j$ is $J_j(\cdot)$.
The Gamma function is $\Gamma(\cdot)$, and the Pochhammer symbol is defined as $(x)_n=\Gamma(x+n)/\Gamma(x)$.
The lower incomplete gamma function is defined as $\gamma(a, x)=\int_0^x t^{a-1}\exp(-t)\,\mathrm{d}t$.
The ramp function is $[x]^+=\max(0,x)$.
The Hermitian (conjugate) transpose is denoted by $(\cdot)^{\mathrm{H}}$.
The inner product of two vectors $\mathbf{x}$ and $\mathbf{y}$ is $\mathbf{x} \cdot \mathbf{y}$.
The ceiling function of a real value $x$ is $\ceil{x}$.
The Kronecker product is~$\otimes$.
The Hadamard product is denoted by~$\odot$.
 The basic 3D rotation matrices about the $x$-, $y$-, and $z$-axes are denoted by $R_x(\cdot)$, $R_y(\cdot)$, and $R_z(\cdot)$, respectively.
$\operatorname{Re}(\cdot)$ and $\operatorname{Im}(\cdot)$ denote the real and imaginary parts of a complex vector, respectively.

\begin{figure}
    \centering
    \includegraphics[width=0.8\columnwidth]{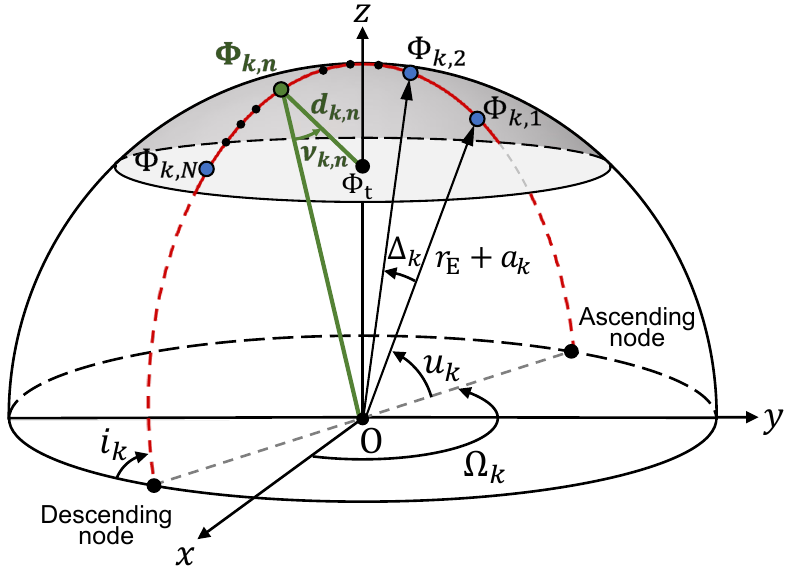}
    \caption{Satellite trajectory defined in an Earth-centered global coordinate system $(x,y,z)$ by a set of orbital elements $\Ok = (\ak,\Omega_k,i_k,u_k)$, where the ground terminal is located at the North Pole ($\phi=90^{\circ}$), i.e., $\Phi_{\ter,n} = [0\;\;0\;\;\re]^{\T}$. Here, $\mathrm{O}$ denotes the Earth’s center, and $\Phi_{k,n}$ represents the position of satellite $k$ at time slot~$n$. The red solid and dashed curves indicate the visible and invisible portions of the satellite orbit with respect to the ground terminal, respectively.
    }
    \label{Fig:system_model}
\end{figure}

\section{System Model}\label{sec:syst_model}
We consider an uplink satellite communication system in which a ground terminal $\ter$ communicates with a serving satellite $\ser$ in the presence of $E$ satellite eavesdroppers $\eve_j$, $j \in \{1,2,\ldots,E\}$. The eavesdroppers are assumed to be non-colluding, i.e., each eavesdropper operates independently without cooperation with others~\cite{my17VTC}. The orbital parameters of potential eavesdropper satellites are assumed to be known, since satellites are physical objects whose orbits are determined by deterministic Keplerian mechanics and can be readily observed and predicted. As illustrated in Fig.~\ref{Fig:system_model}, the satellite orbits are defined in a global Earth-centered coordinate system $(x,y,z)$, where the origin $\mathrm{O}$ is located at the Earth’s center and $\re$ denotes the Earth’s radius. Since the satellites are assumed to follow circular orbits, the orbit of satellite $k$ is characterized by four orbital elements $\Ok = (a_k, \Omega_k, i_k, u_k)$, $k \in \{\ser, \eve_j\}$, representing the orbital altitude (or equivalently the semi-major axis), right ascension of the ascending node (RAAN), inclination, and argument of latitude, respectively. The ground terminal is assumed to be fixed on the Earth’s surface, and its location is specified by the latitude $\phi$ and longitude $\psi$. As shown in Fig.~\ref{Fig:angles}, the terminal is equipped with a uniform planar array (UPA), which is defined in a local coordinate system $(\bar{x},\bar{y},\bar{z})$ centered at the terminal. The UPA consists of $M \triangleq M_{\bar{x}} \times M_{\bar{y}}$ antenna elements, where $M_{\bar{x}}$ and $M_{\bar{y}}$ denote the numbers of antennas along the $\bar{x}$- and $\bar{y}$-axes, respectively.

The terminal transmits information signals only when the serving satellite is within the 3-dB beamwidth of the satellite's receive antenna. This restriction ensures that transmissions occur only when the link quality is sufficiently high, while avoiding unnecessary information leakage when the serving satellite is outside the main beam coverage. To characterize the time-varying satellite geometry, the visible period is discretized into $N_{\mathrm{vis}}$ time slots with a sufficiently small slot interval~$\delta$. The angular evolution of satellite $k \in \{\ser,\eve_j\}$ between adjacent time slots is characterized by an angular offset~$\Deltak$ from the initial argument of latitude $\aolk$, as shown in Fig.~\ref{Fig:system_model}. Accordingly, the argument of latitude of satellite $k$ at time slot $n \in \{1,2,\ldots,N_{\mathrm{vis}}\}$ is given by $\aolkn = \aolk + (n-1)\Deltak$. Since the angular velocity of satellite $k$ is $\omega_k = \sqrt{\frac{G M_{\mathrm{E}}}{(\re + \ak)^3}}$, where $G$ denotes the gravitational constant and $M_{\mathrm{E}}$ is the mass of the Earth~\cite{book13OrbitalVelocity}, the angular offset between adjacent time slots is given by $\Deltak = \omega_k \delta$.

\begin{figure}
    \centering
    \includegraphics[width=0.76\columnwidth]{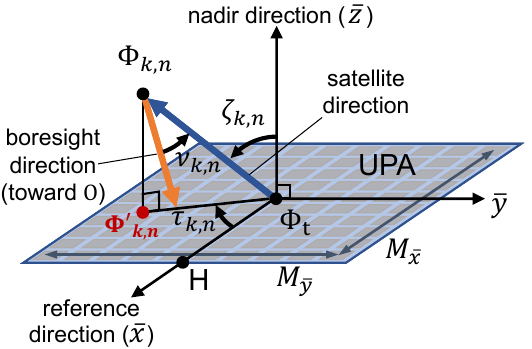}
    \caption{UPA of the ground terminal defined in a local coordinate system $(\bar{x},\bar{y},\bar{z})$, where the terminal is located at the origin. The angles $\nu_{k,n}$, $\zeta_{k,n}$, and $\tau_{k,n}$ denote the off-boresight angle, zenith angle, and azimuth angle associated with satellite $k$ at time slot~$n$, respectively.
    }
    \label{Fig:angles}
\end{figure}

The Earth-centered inertial (ECI) coordinate system is a fundamental reference frame commonly used in satellite and space object tracking. This system is inertial with respect to distant celestial objects, and thus its axes remain fixed without rotating with Earth. In contrast, the Earth-centered, Earth-fixed (ECEF) coordinate system rotates together with Earth. The position of the ground terminal on the Earth’s surface, specified by the latitude $\phi$ and longitude $\psi$, can be expressed in the ECEF coordinate system as
$\Phi_{\ter}^{\mathrm{ECEF}} = \left[ \re \cos \phi \cos \psi,\; \re \cos \phi \sin \psi,\; \re \sin \phi \right]^\mathrm{T}$.
To represent the terminal position in the ECI frame, the Earth’s rotation must be taken into account. Let $\omega_{\mathrm{E}}$ denote the Earth’s rotation rate. In time slot~$n$, the Earth rotates by an angle $\omega_{\mathrm{E}} n \delta$, which can be modeled as a rotation about the $z$-axis using the rotation matrix
\begin{equation}
    R_z(\omega_{\mathrm{E}} n \delta) =
    \begin{bmatrix}
        \cos (\omega_{\mathrm{E}} n \delta) & -\sin (\omega_{\mathrm{E}} n \delta) & 0 \\
        \sin (\omega_{\mathrm{E}} n \delta) & \cos (\omega_{\mathrm{E}} n \delta)  & 0 \\
        0                                   & 0                                    & 1
    \end{bmatrix}.
\end{equation}
Accordingly, the terminal position in the ECI frame at time slot~$n$ is given by
\begin{equation}\label{eq:pos_ter}
    \Phi_{\ter,n}
    =
    R_z(\omega_{\mathrm{E}} n \delta)\Phi_{\ter}^{\mathrm{ECEF}}
    =
    \begin{bmatrix}
        \re \cos\phi \cos \psi_n \\
        \re \cos\phi \sin \psi_n \\
        \re \sin\phi
    \end{bmatrix},
\end{equation}
where $\psi_n = \psi + \omega_{\mathrm{E}} n \delta$.

We assume that the satellites maintain the boresight of their receive beams fixed toward the subsatellite point, i.e., the nearest point on Earth to the satellite~\cite{TR38.821}. Let $\nu_{k,n}$ denote the off-boresight angle of satellite $k$ toward the ground terminal at time slot~$n$, defined as the angle between the terminal direction and the subsatellite point with respect to satellite $k$, as illustrated in Fig.~\ref{Fig:system_model}. Furthermore, let $\nu_k^{\mathrm{3dB}}$ denote the 3-dB beamwidth angle beyond which the receive beam power drops by more than 3 dB. Based on these definitions, the receive antenna gain at satellite $k$ in time slot~$n$ is modeled as~\cite{antBessel}
$G_{k,n} = \Gkmax\left(\frac{J_1(g_{k,n})}{2g_{k,n}} + 36\frac{J_3(g_{k,n})}{g_{k,n}^3}\right)^2$,
where $\Gkmax$ denotes the maximum antenna gain, $J_1(\cdot)$ and $J_3(\cdot)$ are the first- and third-order Bessel functions of the first kind, respectively, and $g_{k,n} \triangleq 2.07123\frac{\sin \nu_{k,n}}{\sin \nu_k^{\mathrm{3dB}}}$. Let $d_{k,n}$ denote the distance between the terminal and satellite $k$ at time slot~$n$. Then, the large-scale path loss between the terminal and satellite $k$ at time slot~$n$ is given by
$\ell_{k,n} = \left(\frac{c}{4\pi f_c}\right)^2 d_{k,n}^{-\kappa}$
  ,
where $c$ is the speed of light, $f_c$ is the carrier frequency, and $\kappa$ is the path-loss exponent.

The characteristics of satellite channels can be accurately modeled by a shadowed-Rician channel model, which explicitly accounts for random shadowing effects. However, its complicated distribution often leads to analytical intractability. Instead, the Nakagami-$m$ fading model is widely adopted as an effective alternative for tractable analysis in satellite communications~\cite{my23WCL}. Thus, we employ the Nakagami-$m$ fading model, which can flexibly capture the dominant line-of-sight (LoS) characteristics between the terminal and satellite~$k$ through the fading parameter $m_k$. Let $\tilde{h}_{k,n}$ denote the small-scale fading coefficient between the terminal and satellite $k$ at time slot $n$. Under the Nakagami-$m$ model, the channel power gain follows the Gamma distribution, i.e., $|\tilde{h}_{k,n}|^2 \sim \Gamma(m_k,\frac{1}{m_k})$. Then, the cumulative distribution function (CDF) of the channel power gain $|\tilde{h}_{k,n}|^2$ is given by
$F_{|\tilde{h}_{k,n}|^2}(x)=\frac{\gamma(m_k,m_k x)}{\Gamma(m_k)}$.

As illustrated in Fig.~\ref{Fig:angles}, we let $\zeta_{k,n}$ and $\tau_{k,n}$ denote the zenith and azimuth angles of satellite $k$ at time slot $n$, respectively, i.e., $(\zeta_{k,n},\tau_{k,n})$ represents the angle-of-departure (AoD) pair from the terminal toward satellite $k$. Based on the UPA structure, the corresponding array response vector toward satellite $k$ at time slot $n$ is given by $\mathbf{a}_{k,n}=\mathbf{a}_{k,n}^{\bar{x}} \otimes \mathbf{a}_{k,n}^{\bar{y}} \in \mathbb{C}^{M}$, where $\mathbf{a}_{k,n}^{\bar{x}} = [1, \:\: e^{-j\pi\sin\zeta_{k,n}\cos\tau_{k,n}}, \:\: \cdots, \:\: e^{-j\pi(M_{\bar{x}}-1)\sin\zeta_{k,n}\cos\tau_{k,n}}]^{\T}$ and $\mathbf{a}_{k,n}^{\bar{y}} = [1, \:\: e^{-j\pi\sin\zeta_{k,n}\sin\tau_{k,n}}, \:\: \cdots, \:\: e^{-j\pi(M_{\bar{y}}-1)\sin\zeta_{k,n}\sin\tau_{k,n}}]^{\T} $
represent the array response vectors along the $\bar{x}$- and $\bar{y}$-axes, respectively. Stacking the eavesdropper array responses column-wise yields $\mathbf{A}_{\eve,n} \triangleq [\mathbf{a}_{\eve_1,n}, \ldots, \mathbf{a}_{\eve_E,n}] \in \mathbb{C}^{M \times E}$.
Consequently, the overall channel vector between the terminal and satellite $k$ at time slot $n$ is modeled as~\cite{my24JSAC_EA}
\begin{align}\label{eq:channel_vector}
    \mathbf{h}_{k,n}=\tilde{h}_{k,n}\sqrt{\ell_{k,n}}\,\mathbf{a}_{k,n},
\end{align}
which incorporates the effects of small-scale fading, large-scale path loss, and array geometry. Let $\wn \in \mathbb{C}^{M}$ denote the beamforming vector employed by the terminal at time slot~$n$, subject to the transmit power constraint $\|\wn\|^2 \leq P_{\mathrm{max}}$, where $P_{\mathrm{max}}$ is the maximum transmit power. Then, the received signal-to-noise ratio (SNR) at satellite $k$ in time slot~$n$ is expressed as
\begin{align}\label{eq:SNR}
    \Gamkn =
    \begin{cases}
        \dfrac{G_{k,n}\left|\hkn^{\H}\wn\right|^2}{N_0 W}, & \text{if } d_{k,n}<d_k^{\mathrm{max}}, \\
        0,                                                 & \text{otherwise},
    \end{cases}
\end{align}
where $d_k^{\mathrm{max}}$ denotes the maximum distance for which satellite~$k$ remains visible due to Earth blockage, given by $d_k^{\mathrm{max}}=\sqrt{\ak(2\re+\ak)}$~\cite{my22TIFS}, $N_0$ denotes the noise power spectral density, and $W$ is the system bandwidth.
For non-colluding eavesdroppers, the secrecy performance is dominated by the most detrimental eavesdropper, i.e., the one achieving the highest received SNR among all eavesdroppers~\cite{my17VTC}. Accordingly, the instantaneous secrecy rate at time slot $n$ is given by
\begin{align}\label{eq:sec_rate}
    R_n=\left[\log_2\left(1+\Gamsn\right)-\log_2\left(1+ \underset{j\in\{1,2,\ldots,E\}}{\max} \Gamen\right)\right]^+.
\end{align}
The secrecy rate in \eqref{eq:sec_rate} is used to formulate the optimization problem in Section \ref{sec:problem_form}.

\section{Mathematical Preliminaries}\label{sec:math_prelim}
In this section, we first characterize the distance between the terminal and satellite $k\in\{\ser,\eve_j\}$, $j\in\{1,2,\cdots,E\}$, in time slot~$n$, i.e., $d_{k,n}$. Then, we analyze the satellite visibility constraint based on the orbital configuration of the satellites.
Additionally, we derive analytical expressions for the off-boresight angle $\nu_{k,n}$, the zenith angle $\zeta_{k,n}$, and the azimuth angle $\tau_{k,n}$.

\subsection{Distance Characterization}
The position of satellite $k\in\{\ser, \eve_j\}$ in time slot~$n$ is obtained by using three successive intrinsic rotations\footnote{In intrinsic rotations, successive rotations are conducted about the axes rotated by the last rotation matrix. The superscript prime ($\prime$) is added to indicate the new axes after an elemental rotation.} with sequence $z-x^{\prime}-z^{\prime\prime}$ as~\cite{my23WCL}
\begin{align}\label{eq:pos_sat}
     & \Phi_{k,n}
    = R_z(\raank)R_x(\inck)R_z(\aolkn) \mu_{x,k}\nonumber               \\
     & = (\re+\ak) \begin{bmatrix}
                       \cos\raank\cos\aolkn-\sin\raank \cos\inck \sin\aolkn \\
                       \cos\raank\cos\inck\sin\aolkn+\sin\raank \cos\aolkn  \\
                       \sin\inck \sin\aolkn
                   \end{bmatrix},
\end{align}
where $\mu_{x,k}=[\re+\ak \:\:\: 0 \:\:\: 0]^{\T}$. From \eqref{eq:pos_ter} and \eqref{eq:pos_sat}, the distance between the terminal and satellite $k\in\{\ser, \eve_j\}$ in time slot~$n$ is given by
\begin{align}\label{eq:dist}
     & d_{k,n}= \Vert \Phi_{k,n}-\Phi_{\ter,n} \Vert \nonumber                                                \\
     & =\Big(
    (\re + a_k)^2 \Big(\eta(\cdot;i_k^2,u_{k,n}^2)+ \eta(\Omega_k^2,i_k^2;u_{k,n}^2)\nonumber                 \\
     & +
    \eta(u_{k,n}^2;\Omega_k^2)
    +
    \eta(\Omega_k^2,u_{k,n}^2;\cdot)
    +\eta(i_k^2;\Omega_k^2,u_{k,n}^2)  \Big) \nonumber                                                        \\
     & - 2(\re + a_k) \re \Big(\eta(\Omega_k,i_k,\phi;u_{k,n},\psi_n) -\eta(i_k,\phi,\psi_n;\Omega_k,u_{k,n})
    \nonumber                                                                                                 \\
     & + \eta(\Omega_k,u_{k,n},\phi,\psi_n;\cdot) +
    \eta(u_{k,n},\phi;\Omega_k,\psi_n)
    +
    \eta(\cdot;i_k,u_{k,n},\phi) \Big)  \nonumber                                                             \\
     & + \re^2 \Big(
    \eta(\phi^2,\psi_n^2;\cdot) +
    \eta(\phi^2;\psi_n^2) +
    \eta(\cdot;\phi^2) \Big) \Big)^{1/2},
\end{align}
where $\eta(p_1^{o_1},\cdots,p_{\mathrm{A}}^{o_\mathrm{A}};q_1^{v_1},\cdots,q_{\mathrm{B}}^{v_\mathrm{B}})
    \delequal \prod_{t=1}^{\mathrm{A}}\cos^{o_t}{p_t}\times \prod_{t=1}^{\mathrm{B}}\sin^{v_t}{q_t}$ is a multiplication of cosine and sine functions.
These distances vary over time due to satellite mobility and directly affect the path loss $\ell_{k,n}$.

\subsection{Satellite Visibility Analysis}
Since the orbital plane of the serving satellite is obtained by two intrinsic rotations with orbital parameters $\raans$ and $\incs$, the normal vector of the orbital plane is obtained as
\begin{align}
    \mathbf{n}_{\ser}^{\perp}
    = R_z(\raans)R_x(\incs) \hat{z}
    = \begin{bmatrix}
          \sin \raans \sin \incs  \\
          -\cos \raans \sin \incs \\
          \cos \incs
      \end{bmatrix},
\end{align}
where $\hat{z}=[0 \:\: 0 \:\: 1]^{\T}$.
The angle between the terminal and the normal vector is obtained as
\begin{align}
    \beta^{\prime}
     & =\arccos \frac{\mathbf{n}_{\ser}^{\perp} \cdot \Phi_{\ter,n}}{\lVert \mathbf{n}_{\ser}^{\perp} \rVert \cdot \lVert \Phi_{\ter,n} \rVert}\nonumber \\
     & = \arccos (\cos \phi \cos \psi_n \sin \raans \sin \incs \nonumber                                                                                 \\
     & \quad - \cos \phi \sin \psi_n \cos \raans \sin \incs + \sin \phi \cos \incs),
\end{align}
and the angle between the terminal and the orbital plane is derived as
\begin{align}
    \beta=\Big\vert\beta^{\prime}-\frac{\pi}{2}\Big\vert.
\end{align}
According to~\cite{my24TWC}, the orbital plane is visible only when the following criterion holds:
\begin{align}\label{eq:vis_criterion}
    \beta  < \arccos\left(\frac{\re}{\re+\as}\right),
\end{align}
and the length of the visible arc of the orbit is given by
\begin{align}\label{eq:len_vis_arc}
    l
    =2(\re+\as)\arcsin\left(\sqrt{1-\frac{\sec^2\beta}{(1+\as/\re)^2}}\right).
\end{align}
\begin{rem}
    The orbital visibility depends on the terminal position and the orbital plane configuration as $\beta$ is determined by the terminal's latitude $\phi$ and longitude $\psi_n$, and the orbital elements $\raans$ and $\incs$.
\end{rem}
\begin{rem}
    The orbital visibility constraint, i.e., $\arccos\left(\frac{\re}{\re+\as}\right)$, increases with the altitude $\as$, which means that better orbital visibility could be achieved at a higher altitude because the satellite can be seen from a wider range of terminal positions.
\end{rem}
\begin{rem}
    The length of the visible arc increases as $\beta$ decreases, and reaches the maximum when $\beta$ becomes zero. This explains that the terminal achieves the highest visibility when the terminal is located on the orbital plane.
\end{rem}
Using \eqref{eq:len_vis_arc}, the angle between the two endpoints of the visible arc is obtained as
\begin{align}\label{eq:thetas}
    \Theta
    = \frac{l}{\re+\as}=2\arcsin\left(\sqrt{1-\frac{\sec^2\beta}{(1+\as/\re)^2}}\right).
\end{align}
From \eqref{eq:thetas} and the angular offset $\Deltas$ between consecutive time slots, the maximum number of visible time slots is obtained as
\begin{align}
    N_{\mathrm{vis}}
     & = \left\lceil \frac{\Theta}{\Deltas} \right\rceil
    = \left\lceil \frac{l}{(\re+\as)\omega_{\ser}\delta} \right\rceil \nonumber                                                                          \\
     & = \left\lceil \sqrt{\frac{4(\re + \as)^3}{\delta^2 G M_{\mathrm{E}}}}\arcsin\left(\sqrt{1-\frac{\sec^2\beta}{(1+\as/\re)^2}}\right) \right\rceil.
\end{align}

\begin{rem}
    The maximum number of visible time slots increases with the altitude of satellites because the feasible region satisfying the visibility constraint \eqref{eq:vis_criterion} enlarges. For example, when $\as=\{300, 600, 1200\}$ km, the terminal is visible under $\beta$ less than $\{17.2, 23.9, 32.7\}$ degrees. This indicates that higher altitudes are preferable when satellite visibility is important, even though it comes at the cost of reduced signal quality due to larger path loss.
\end{rem}

Now, we derive the visible range of the argument of latitude~$u_{\ser}$ based on the fact that the maximum elevation angle is achieved at the position in orbit with the minimum distance to the terminal.

\begin{lem}\label{lem:aol_max_elev_angle}
    The argument of latitude corresponding to the maximum elevation angle, i.e., $\aols + \frac{\Theta}{2}$, equals the angle between the ascending node and the terminal position projected onto the orbital plane.
\end{lem}

\begin{IEEEproof}
    As illustrated in Fig.~\ref{Fig:max_elev_angle}, the terminal position projected onto the orbital plane is expressed as
    \begin{align}
        \bar{\Phi}_{\ter,n} = {\Phi}_{\ter,n} - ({\Phi}_{\ter,n} \cdot \mathbf{n}_{\ser}^{\perp})\mathbf{n}_{\ser}^{\perp}.
    \end{align}
    To obtain the angle between $\bar{\Phi}_{\ter,n}$ and the ascending node, we define the two basis vectors $\mathbf{e}_{x}^{\mathrm{orb}}$ and $\mathbf{e}_{y}^{\mathrm{orb}}$ of the orbital plane through basis rotation as
    \begin{align}
         & \mathbf{e}_{x}^{\mathrm{orb}}=R_z(\raans)R_x(\incs)\hat{x}=
        \begin{bmatrix}
            \cos\raans \\
            \sin\raans \\
            0
        \end{bmatrix},                           \\
         & \mathbf{e}_{y}^{\mathrm{orb}}=R_z(\raans)R_x(\incs)\hat{y}=
        \begin{bmatrix}
            -\sin\raans\cos\incs \\
            \cos\raans\cos\incs  \\
            \sin\incs
        \end{bmatrix}.
    \end{align}
    Using these basis vectors, we obtain the argument of latitude corresponding to the maximum elevation angle as $\arctan \left( \frac{\bar{\Phi}_{\ter,n} \cdot \mathbf{e}_{y}^{\mathrm{orb}}}{\bar{\Phi}_{\ter,n} \cdot \mathbf{e}_{x}^{\mathrm{orb}}}\right)$.
    As this angle equals $\aols + \frac{\Theta}{2}$, the initial argument of latitude for the visible arc is given by
    \begin{align}
        \aols
        = \arctan \left( \frac{\bar{\Phi}_{\ter,n} \cdot \mathbf{e}_{y}^{\mathrm{orb}}}{\bar{\Phi}_{\ter,n} \cdot \mathbf{e}_{x}^{\mathrm{orb}}}\right) - \frac{\Theta}{2}.
    \end{align}
    Thus, the visible range of the argument of latitude is from $\aols$ to $\aols + \Theta$.
\end{IEEEproof}

\begin{figure}
    \centering
    \includegraphics[width=0.8\columnwidth]{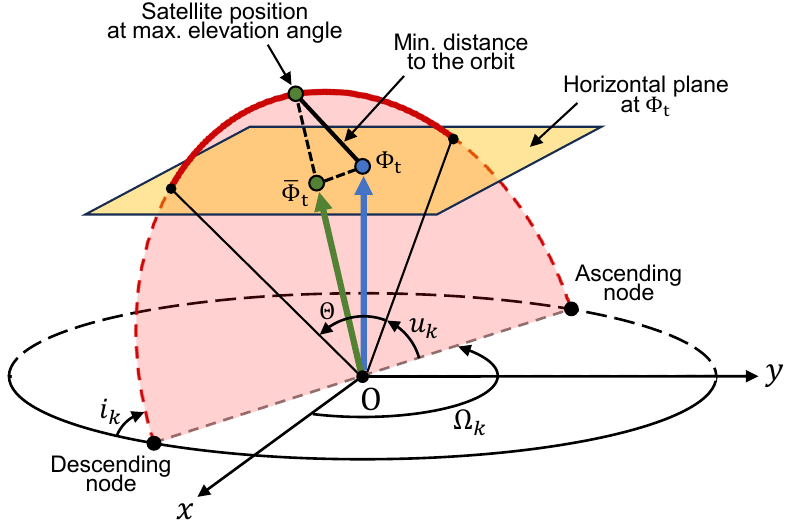}
    \caption{Satellite position at the maximum elevation angle. The red solid and dashed lines represent the visible and invisible parts of the LEO orbit.}
    \label{Fig:max_elev_angle}
\end{figure}

\subsection{Geometric Pointing Angle Characterization}
In this subsection, we derive analytical expressions for the off-boresight angle $\nu_{k,n}$, the zenith angle $\zeta_{k,n}$, and the azimuth angle $\tau_{k,n}$ of satellite $k$.
The off-boresight angle $\nu_{k,n}$ is defined as the angle between the boresight direction and the terminal with respect to satellite $k$.
To derive the off-boresight angle of satellite $k$, we apply the law of cosines to the triangle $\triangle \mathrm{O} \Phi_{\ter,n} \Phi_{k,n}$,
i.e., $\overline{\mathrm{O}\Phi_{\ter,n}}^2
    =\overline{\Phi_{\ter,n}\Phi_{k,n}}^2+\overline{\mathrm{O}\Phi_{k,n}}^2-2\overline{\Phi_{\ter,n}\Phi_{k,n}}\cdot\overline{\mathrm{O}\Phi_{k,n}} \cdot\cos \nu_{k,n}$.
Since $\overline{\mathrm{O}\Phi_{\ter,n}} = \re$, $\overline{\Phi_{\ter,n}\Phi_{k,n}} = d_{k,n}$, and $\overline{\mathrm{O}\Phi_{k,n}} = \re + a_k$, the off-boresight angle is expressed as
\begin{align}
    \nu_{k,n}
    =
    \arccos\left(\frac{d_{k,n}^2 + (\re + a_k)^2 - \re^2}{2 d_{k,n} (\re + a_k)}\right).
\end{align}

The zenith angle $\zeta_{k,n}$ is defined as the angle between the local vertical (zenith) direction at the terminal and the direction toward satellite $k$.
It is calculated as
\begin{align}
    \zeta_{k,n} = \arccos\left(\frac{\Phi_{\ter,n} \cdot (\Phi_{k,n} - \Phi_{\ter,n})}{\lVert \Phi_{\ter,n} \rVert \cdot \lVert \Phi_{k,n} - \Phi_{\ter,n} \rVert}\right).
\end{align}

To derive the azimuth angle $\tau_{k,n}$ with respect to the terminal, we first define the equation of the horizontal plane at the terminal's position $\Phi_{\ter,n}$ as $\Phi_{\ter,n} \cdot ([x, y, z]^{\T} - \Phi_{\ter,n}) = 0$, which is equivalent to
\begin{align}\label{eq:horizontal_plane}
     & (\cos \phi \cos \psi_n) x + (\cos \phi \sin \psi_n) y + (\sin \phi) z = \re.
\end{align}
This plane corresponds to the $\bar{x} \bar{y}$-plane in the local coordinate system ($\bar{x}$, $\bar{y}$, $\bar{z}$) where the terminal's antennas are located.
Let $\Phi^{\prime}_{k,n}$ denote the projection of $\Phi_{k,n}$ onto this horizontal plane, which is given by
\begin{align}
    \Phi^{\prime}_{k,n} = \Phi_{k,n} - \frac{\Phi_{\ter,n} \cdot (\Phi_{k,n} - \Phi_{\ter,n})}{\| \Phi_{\ter,n} \|^2} \Phi_{\ter,n}.
\end{align}
To determine the reference direction for the planar array, we use an arbitrary point $\mathrm{H}$ on the horizontal plane \eqref{eq:horizontal_plane}, as illustrated in Fig.~\ref{Fig:angles}.
This reference direction indicates the orientation of the planar array on the horizontal plane at the terminal.
Without loss of generality, we set $\mathrm{H} = \left[\frac{\re}{\cos \phi \cos \psi_n}, 0, 0\right]^{\T}$.
Thus, the azimuth angle is obtained as
\begin{align}
    \tau_{k,n} = \arccos\left(\frac{(\Phi^{\prime}_{k,n} - \Phi_{\ter,n}) \cdot (\mathrm{H} - \Phi_{\ter,n})}{\|\Phi^{\prime}_{k,n} - \Phi_{\ter,n}\| \cdot \|\mathrm{H} - \Phi_{\ter,n}\|}\right).
\end{align}
\noindent The above $\arccos$-based formulation illustrates the geometric relationship for general terminal latitudes where the reference point~$\mathrm{H}$ is well-defined.
The angles $\nu_{k,n}$, $\zeta_{k,n}$, and $\tau_{k,n}$ determine $G_{k,n}$ and $\mathbf{a}_{k,n}$, both of which vary with the time slot~$n$.

\section{Outage-Constrained Secrecy Rate Maximization Problem}\label{sec:problem_form}

In this section, we first formulate the secrecy rate maximization problem subject to transmit power and probabilistic outage constraints over the transmission slots. We then derive closed-form expressions for the connection and secrecy outage probabilities under the Nakagami-$m$ fading. Finally, we transform the probabilistic constraints into deterministic beam gain constraints, yielding a tractable problem reformulation that serves as the basis for the RL-based solution developed in Section~\ref{sec:rl_solution}.


As discussed in Section~\ref{sec:syst_model}, the terminal transmits uplink signals only when the serving satellite is within the 3-dB beamwidth, i.e., $\nu_{\mathrm{ser},n} < \nu_{\mathrm{ser}}^{\mathrm{3dB}}$. Therefore, the number of transmission slots, denoted by $N$, must always be less than the number of visible slots $N_{\mathrm{vis}}$, i.e., $N \leq N_{\mathrm{vis}}$.
The objective is to maximize the expected secrecy rate while satisfying the transmit power constraint and average outage probability constraints over the $N$ transmission slots. The secrecy rate maximization problem is formulated as follows:
\begin{align}
    (\mathrm{P1}) \quad \underset{\{\mathbf{w}_n\}_{n=1}^{N}}{\text{maximize}} \quad &
    \mathbb{E}\left[\sum_{n=1}^{N} R_n(\mathbf{w}_n)\right] \label{eq:rate_obj}                                                                                                                   \\
    \text{subject to} \quad
                                                                                     & \|\mathbf{w}_n\|^2 \leq P_{\mathrm{max}}, \quad \forall n, \label{eq:inst_power_const}                     \\
                                                                                     & \frac{1}{N}\sum_{n=1}^{N} P_n^{\mathrm{co}}(\mathbf{w}_n) \leq \epsilon_{\mathrm{co}}, \label{eq:co_const} \\
                                                                                     & \frac{1}{N}\sum_{n=1}^{N} P_n^{\mathrm{so}}(\mathbf{w}_n) \leq \epsilon_{\mathrm{so}}, \label{eq:so_const}
\end{align}
where the expectation in \eqref{eq:rate_obj} is over the small-scale fading realizations, $P_n^{\mathrm{co}}(\mathbf{w}_n)$ denotes the connection outage probability at slot $n$, which is the probability that the serving satellite cannot decode the message, $P_n^{\mathrm{so}}(\mathbf{w}_n)$ denotes the secrecy outage probability at slot $n$, defined as the probability that at least one eavesdropper can intercept the message, $\epsilon_{\mathrm{co}} \in (0,1)$ is the maximum tolerable average connection outage probability, and $\epsilon_{\mathrm{so}} \in (0,1)$ is the maximum tolerable average secrecy outage probability. The average formulation \eqref{eq:co_const} and \eqref{eq:so_const} is well-suited to the transmission slots because the terminal’s objective is to maintain reliable and secure communication on average over the transmission slots. This allows the beamforming policy to allocate resources adaptively across time slots with varying channel conditions.
Closed-form expressions for $P_n^{\mathrm{co}}$ and $P_n^{\mathrm{so}}$ are derived in the following subsection to enable tractable evaluation of the constraints \eqref{eq:co_const} and \eqref{eq:so_const}.

\subsection{Outage Probability Analysis}\label{sec:outage_analysis}

To characterize the per-slot outage probabilities appearing in \eqref{eq:co_const} and \eqref{eq:so_const}, we derive closed-form expressions under the Nakagami-$m$ fading assumption.

\subsubsection{Connection Outage Probability}
The connection outage probability is defined as the probability that the instantaneous data rate at the serving satellite $\ser$ falls below a target rate $\Rsn$. Mathematically, it is given by
\begin{align}\label{eq:Pco1}
    P_n^{\mathrm{co}}
     & = \mathbb{P}\left(\log_2(1+\Gamma_{\ser,n}) < \Rsn\right)\nonumber                                                                                                                              \\
     & =\mathbb{P}\left( \big| \mathbf{h}_{\ser,n}^{\mathrm{H}} \mathbf{w}_n \big|^2 < \frac{(2^{\Rsn} - 1) N_0 W}{G_{\ser,n}} \right)\nonumber                                                        \\
     & \mathop=^{(a)}\mathbb{P}\left( |\tilde{h}_{\ser,n}|^2 < \frac{(2^{\Rsn} - 1) N_0 W}{G_{\ser,n} \ell_{\ser,n} \big| \mathbf{a}_{\ser,n}^{\mathrm{H}} \mathbf{w}_n \big|^2} \right) \nonumber     \\
     & \mathop=^{(b)} \frac{1}{\Gamma(m_{\ser})}\gamma\left(m_{\ser}, \frac{m_{\ser}(2^{\Rsn} - 1) N_0 W}{G_{\ser,n} \ell_{\ser,n} \big| \mathbf{a}_{\ser,n}^{\mathrm{H}} \mathbf{w}_n \big|^2}\right),
\end{align}
where ($a$) follows from the channel decomposition $|\mathbf{h}_{\ser,n}^{\mathrm{H}} \mathbf{w}_n|^2 = \ell_{\ser,n} |\tilde{h}_{\ser,n}|^2 |\mathbf{a}_{\ser,n}^{\mathrm{H}} \mathbf{w}_n|^2$, and ($b$) follows from the CDF of the Gamma distribution.

\subsubsection{Secrecy Outage Probability}
The individual secrecy outage probability for the $j$-th eavesdropper, denoted $P_{j,n}^{\mathrm{so}}$, is defined as the probability that the instantaneous achievable rate at that eavesdropper exceeds a secrecy threshold $R_{\eve_j,n}$. 
Here, the threshold $R_{\eve_j,n}$ represents the redundancy rate allocated through wiretap coding to protect the confidential information, rather than the data rate intended for the eavesdropper or the directly tolerable leakage rate. Therefore, a secrecy outage occurs when the eavesdropper’s instantaneous channel capacity exceeds this redundancy margin, indicating that the confidential message may no longer be fully protected from information leakage.
Similar to the derivation of \eqref{eq:Pco1}, the closed-form expression is obtained as
\begin{align}\label{eq:Pso_j}
    \Pson
     & =\mathbb{P}(\log_2(1+\Gamma_{\eve_j,n}) > R_{\eve_j,n})\nonumber                                                                                                                                                 \\
     & =1-\frac{1}{\Gamma(m_{\eve_j})}\gamma\left(m_{\eve_j}, \frac{m_{\eve_j}(2^{R_{{\eve_j},n}} - 1) N_0 W}{G_{{\eve_j},n} \ell_{{\eve_j},n} \big| \mathbf{a}_{{\eve_j},n}^{\mathrm{H}} \mathbf{w}_n \big|^2}\right).
\end{align}
Under the assumption that the eavesdroppers operate independently without collusion, a secrecy outage event occurs when at least one of the $E$ eavesdroppers successfully intercepts the transmitted message. Since the fading channels across different eavesdroppers are statistically independent, due to their physical separation across distinct orbital planes specified by $\mathcal{O}_{\eve_j} = (a_{\eve_j}, \Omega_{\eve_j}, i_{\eve_j}, u_{\eve_j})$, the probability that no eavesdropper succeeds is given by the product of individual complement probabilities. The overall secrecy outage probability is therefore expressed as
\begin{align}\label{eq:Pso_all}
    P^{\mathrm{so}}_n
     & = 1 - \mathbb{P}\left(
    \bigcap_{j=1}^{E} \left\{
    \log_2\left(1+\Gamma_{{\eve_j},n}\right)
    \leq R_{{\eve_j},n}
    \right\}\right) \nonumber                                                          \\
     & = 1 -
    \prod_{j=1}^{E}
    \mathbb{P}\left(\log_2(1+\Gamma_{{\eve_j},n}) \leq R_{{\eve_j},n}\right) \nonumber \\
     & = 1 -
    \prod_{j=1}^{E}
    \left(1 - P^{\mathrm{so}}_{j,n}\right),
\end{align}
where the second equality holds due to the independence of the fading channels across different eavesdroppers.
From \eqref{eq:Pso_j} and \eqref{eq:Pso_all}, $P_n^{\mathrm{so}}$ is monotonically increasing in the beamforming gain $|\mathbf{a}_{\eve_j,n}^{\mathrm{H}} \mathbf{w}_n|^2$ toward each eavesdropper $j$. This observation motivates the design of beamforming vectors that suppress signal leakage toward potential eavesdroppers.

\subsection{Probabilistic Constraint Reformulation}\label{sec:det_approx}

The outage probability constraints in \eqref{eq:co_const} and \eqref{eq:so_const} involve the lower incomplete gamma function, which does not admit a closed-form expression amenable to direct optimization. To address this difficulty, a conservative bound on the regularized lower incomplete gamma function is employed, as stated in the following lemma.

\begin{lem}\label{lem:gamma_approx}
    For $m \geq 1$ and $x \geq 0$, the regularized lower incomplete gamma function satisfies Alzer's inequalities~\cite{Alzer}, expressed as
    \begin{align}\label{eq:gamma_approx}
        \left(1 - e^{-(m!)^{-1/m} x}\right)^m \leq \frac{\gamma(m, x)}{\Gamma(m)} \leq \left(1 - e^{-x}\right)^m.
    \end{align}
\end{lem}

Applying the upper bound from Lemma~\ref{lem:gamma_approx} to \eqref{eq:Pco1}, we obtain the upper bound of the connection outage probability as
\begin{align}\label{eq:Pco_bound}
    P_n^{\mathrm{co}} = \frac{\gamma(m_{\ser}, x_{\ser,n})}{\Gamma(m_{\ser})} \leq \left(1 - e^{-x_{\ser,n}}\right)^{m_{\ser}} \triangleq c_{\mathrm{co},n},
\end{align}
where $x_{\ser,n} \delequal \frac{m_{\ser} (2^{\Rsn} - 1) N_0 W}{G_{\ser,n} \ell_{\ser,n} |\mathbf{a}_{\ser,n}^{\mathrm{H}} \mathbf{w}_n|^2}$.
For the secrecy outage, we apply the lower bound in \eqref{eq:gamma_approx} to each eavesdropper's CDF in \eqref{eq:Pso_j}, which gives $1 - P_{j,n}^{\mathrm{so}} \geq (1 - e^{-(m_{\eve_j}!)^{-1/m_{\eve_j}} x_{\eve_j,n}})^{m_{\eve_j}}$. Substituting this into \eqref{eq:Pso_all} yields the upper bound of the secrecy outage probability, i.e.,
\begin{align}\label{eq:Pso_bound}
    P_n^{\mathrm{so}} \leq 1 - \prod_{j=1}^{E} \left(1 - e^{-(m_{\eve_j}!)^{-1/m_{\eve_j}} x_{\eve_j,n}}\right)^{m_{\eve_j}} \triangleq c_{\mathrm{so},n}.
\end{align}
where $x_{\eve_j,n} \delequal \frac{m_{\eve_j} (2^{R_{\eve_j,n}} - 1) N_0 W}{G_{\eve_j,n} \ell_{\eve_j,n} |\mathbf{a}_{\eve_j,n}^{\mathrm{H}} \mathbf{w}_n|^2}$, analogous to $x_{\ser,n}$.
Since $P_n^{\mathrm{co}} \leq c_{\mathrm{co},n}$ and $P_n^{\mathrm{so}} \leq c_{\mathrm{so},n}$ at every time slot, the averages satisfy $\frac{1}{N}\sum_n P_n^{\mathrm{co}} \leq \frac{1}{N}\sum_n c_{\mathrm{co},n}$ and $\frac{1}{N}\sum_n P_n^{\mathrm{so}} \leq \frac{1}{N}\sum_n c_{\mathrm{so},n}$. Therefore, imposing constraints on the average of $c_{\mathrm{co},n}$ and $c_{\mathrm{so},n}$ constitutes a sufficient condition for the original average outage constraints \eqref{eq:co_const} and \eqref{eq:so_const}. The connection and secrecy outage bounds exhibit asymmetric tightness, yet both remain conservative.
The resulting cost functions are closed-form expressions of elementary functions, readily amenable to gradient-based policy optimization in the CMDP framework developed in the following section.

\section{Proposed RL-Based Beamformer Design}\label{sec:rl_solution}

The optimization problem (P1) is non-convex due to the non-concave secrecy rate objective and the dependence between beamforming vectors across time slots.
Iterative methods such as SCA face two issues in this setting. First, the per-slot cost is high, since each slot requires repeatedly linearizing the non-concave secrecy rate and solving a dimension-$2M$ quadratic program from multiple initializations. Second, the constraints are handled per-slot rather than as time-averaged budgets. We therefore use SCA only as an offline benchmark in Section~\ref{sec:sim_res}.

To address these limitations, we adopt an RL-based approach in which the beamforming vector is obtained by a single forward pass of the policy network, eliminating the per-slot iterative optimization.  The problem is reformulated as a CMDP and solved via a PD-SAC algorithm.

\subsection{CMDP Formulation}\label{sec:cmdp}

As discussed in Section~\ref{sec:syst_model}, the terminal transmits only when it is located within the serving satellite's coverage region, i.e., $\nu_{\ser,n} < \nu_{\ser}^{\mathrm{3dB}}$.
In the RL framework, the agent acts only during these transmission slots, and the remaining slots are skipped without agent interaction.
The sequential beamforming optimization is modeled as a CMDP defined by the tuple $(\mathcal{S}, \mathcal{A}, P, r, \mathbf{c}, \boldsymbol{\epsilon})$.
Here, $\mathcal{S}$ is the state space, $\mathcal{A}$ is the action space, $P: \mathcal{S} \times \mathcal{A} \times \mathcal{S} \to [0,1]$ is the transition probability function, and $r: \mathcal{S} \times \mathcal{A} \to \mathbb{R}$ is the per-step reward function. The vector $\mathbf{c} = [c_{\mathrm{co}}, c_{\mathrm{so}}]^{\mathrm{T}}$ collects the per-step cost functions for connection and secrecy outage, with realization $\mathbf{c}_n = [c_{\mathrm{co},n}, c_{\mathrm{so},n}]^{\mathrm{T}}$ at slot $n$, and $\boldsymbol{\epsilon} = [\epsilon_{\mathrm{co}}, \epsilon_{\mathrm{so}}]^{\mathrm{T}}$ is the vector of corresponding cost thresholds.

\subsubsection{State Space}
The state $s_n \in \mathcal{S}$ at time slot $n$ captures the geometry-derived channel characteristics, the previous-slot secrecy rate, and temporal position within the transmission slots, i.e., 
\begin{align}\label{eq:state_space}
    s_n = \bigl\{ & \bar{R}_{n-1},\, \tfrac{n}{N},\, \tilde{\ell}_{\ser,n},\, \tilde{G}_{\ser,n},\, \tilde{\boldsymbol{\ell}}_{\eve,n},\, \tilde{\mathbf{G}}_{\eve,n},\, \nonumber                \\
                  & \operatorname{Re}(\mathbf{a}_{\ser,n}),\, \operatorname{Im}(\mathbf{a}_{\ser,n}),\, \operatorname{Re}(\mathbf{A}_{\eve,n}),\, \operatorname{Im}(\mathbf{A}_{\eve,n}) \bigr\}.
\end{align}
Here, $\bar{R}_{n-1} \approx [\log_2(1 + \bar{\Gamma}_{\ser,n-1}) - \log_2(1 + \max_j \bar{\Gamma}_{\eve_j,n-1})]^+$  is the long-term secrecy rate from the previous time slot with $\bar{\Gamma}_{k,n-1} = G_{k,n-1} \ell_{k,n-1} |\mathbf{a}_{k,n-1}^{\mathrm{H}} \mathbf{w}_{n-1}|^2 / (N_0 W)$. Since accurate instantaneous CSI is difficult to obtain in practice for satellite links, $\bar{R}_{n-1}$ is constructed solely from the satellite orbital geometry and the previously applied beamforming vector, both of which are known to the terminal, thereby providing a deterministic estimate of the achievable secrecy performance without requiring real-time channel estimation. The term $n/N \in [0,1]$ is the normalized time-slot index, indicating its relative position in time. The normalized path losses are defined as $\tilde{\ell}_{k,n} = (10\log_{10}(\ell_{k,n}) + 200)/20$ to scale the values near zero for improved neural network training stability, with $\tilde{\boldsymbol{\ell}}_{\eve,n} = [\tilde{\ell}_{\eve_1,n}, \ldots, \tilde{\ell}_{\eve_{E},n}]^{\mathrm{T}}$.
The normalized antenna gains $\tilde{G}_{k,n} = 10\log_{10}(G_{k,n})/G_{\max}^{\mathrm{dBi}}$, where $G_{\max}^{\mathrm{dBi}} \triangleq 10\log_{10}(\Gkmax)$ is the peak antenna gain in dBi, reflect the angular gain attenuation toward each satellite, with $\tilde{\mathbf{G}}_{\eve,n} = [\tilde{G}_{\eve_1,n}, \ldots, \tilde{G}_{\eve_{E},n}]^{\mathrm{T}}$. Concatenating all components and splitting the complex array responses into real and imaginary parts yields a real-valued state vector of dimension $D_\mathrm{s} = 4 + 2E + 2M + 2EM$.

\subsubsection{Action Space}
To enforce the power constraint $\|\mathbf{w}_n\|^2 \leq P_{\mathrm{max}}$ while keeping the actor network differentiable, the actor outputs a latent action that consists of the beamforming direction $\tilde{q}_{n}^{\mathrm{dir}} \in \mathbb{R}^{2M}$ and the transmit power $\tilde{q}_{n}^{\mathrm{pow}} \in \mathbb{R}$, and is deterministically mapped to the beamforming vector.
Here, $\tilde{q}_{n}^{\mathrm{dir}}$ stacks the real and imaginary parts of $\mathbf{w}_n \in \mathbb{C}^{M}$, a Cartesian representation chosen to avoid the discontinuity of cyclic phase variables.
These two components are concatenated into the latent action vector $\tilde{q}_n = [(\tilde{q}_{n}^{\mathrm{dir}})^{\mathrm{T}},\, \tilde{q}_{n}^{\mathrm{pow}}]^{\mathrm{T}} \in \mathbb{R}^{2M+1}$. The direction components are bounded by $\tanh(\cdot)$ and normalized to a unit vector, while the power component is mapped to $(0, P_{\mathrm{max}})$ via the sigmoid function $\sigma(\cdot)$ as follows:
\begin{align}
    \hat{\mathbf{d}}_n & = \frac{\tanh(\tilde{q}_{n}^{\mathrm{dir}})}{\|\tanh(\tilde{q}_{n}^{\mathrm{dir}})\|}, \label{eq:action_dir} \\
    P_n                & = P_{\mathrm{max}} \cdot \sigma(\tilde{q}_{n}^{\mathrm{pow}}), \label{eq:action_power}                       \\
    \mathbf{w}_n       & = \sqrt{P_n}\, \hat{\mathbf{d}}_n. \label{eq:action_norm}
\end{align}
With a slight abuse of notation, $\mathbf{w}_n$ denotes either the complex beamforming vector $\mathbf{w}_n \in \mathbb{C}^{M}$ from  Section~\ref{sec:syst_model} or its real-valued representation $[\mathrm{Re}(\mathbf{w}_n)^{\mathrm{T}}, \mathrm{Im}(\mathbf{w}_n)^{\mathrm{T}}]^{\mathrm{T}} \in \mathbb{R}^{2M}$ used as the neural network input. The critic networks receive $\mathbf{w}_n \in \mathbb{R}^{2M}$. The policy samples the latent action $\tilde{q}_n \in \mathbb{R}^{2M+1}$ from a diagonal Gaussian, and its squashed counterpart $(u_n, v_n) = (\tanh(\tilde{q}_n^{\mathrm{dir}}),\, P_{\mathrm{max}}\sigma(\tilde{q}_n^{\mathrm{pow}}))$ defines the policy density and entropy used in the SAC update.

\subsubsection{Reward and Cost Functions}
The objective is to maximize the expected cumulative secrecy rate. The per-step reward is the instantaneous secrecy rate defined in \eqref{eq:sec_rate}, which includes the small-scale fading realization: $r_n(s_n, \mathbf{w}_n) = R_n(\mathbf{w}_n).$
Note that $\bar{R}_{n-1}$ uses only the deterministic geometry, whereas the reward $R_n$ reflects the actual channel realization including small-scale fading.
Since the policy is trained offline through simulation, the reward can be computed using the actual channel realization generated in the simulator, including small-scale fading.
The per-step cost functions are the upper-bound cost functions $c_{\mathrm{co},n}$ and $c_{\mathrm{so},n}$ derived in Section~\ref{sec:det_approx}. The power consumption is not modeled as a cost function, since the sigmoid parameterization \eqref{eq:action_power} ensures $\|\mathbf{w}_n\|^2 = P_n \leq P_{\mathrm{max}}$ at every transmission slot.

\subsubsection{Optimization Objective}
The CMDP objective is to find a policy $\pi: \mathcal{S} \to \mathcal{P}(\mathcal{A})$ that maximizes the expected cumulative reward $J_{\mathrm{R}}$ while satisfying the average cost constraints $J_{\mathrm{co}}$ and $J_{\mathrm{so}}$
\cite{book99_cmdp}:
\begin{align}
    (\mathrm{P2}) \quad \underset{\pi}{\text{maximize}} \quad & J_{\mathrm{R}}(\pi) = \mathbb{E}_{\pi} \left[ \sum_{n=1}^{N} r_n \right] \label{eq:RL_obj}                                                                    \\*
    \text{subject to} \quad                                         & J_{\mathrm{co}}(\pi) = \mathbb{E}_{\pi} \left[ \frac{1}{N} \sum_{n=1}^{N} c_{\mathrm{co},n} \right] \leq \epsilon_{\mathrm{co}}, \label{eq:RL_const_co} \\*
                                                              & J_{\mathrm{so}}(\pi) = \mathbb{E}_{\pi} \left[ \frac{1}{N} \sum_{n=1}^{N} c_{\mathrm{so},n} \right] \leq \epsilon_{\mathrm{so}}, \label{eq:RL_const_so}
\end{align}
where $\mathbb{E}_{\pi}$ denotes the expectation over the trajectories induced by the policy $\pi$. As discussed in Section~\ref{sec:det_approx}, satisfying \eqref{eq:RL_const_co} and \eqref{eq:RL_const_so} also ensures the original average outage constraints \eqref{eq:co_const} and \eqref{eq:so_const}. 
The formulation is undiscounted, i.e., $\gamma_{\mathrm{d}} = 1$. The finite horizon $N$ bounds the cumulative reward and cost, and the undiscounted objective \eqref{eq:RL_obj} matches the original problem \eqref{eq:rate_obj} exactly.

\subsection{PD-SAC Algorithm}\label{sec:pd_sac}

We solve the CMDP using Lagrangian relaxation combined with the SAC framework~\cite{achiam17_cpo,19_RCPO,paternain23_safe,yang21_wcsac,haarnoja18_sac,haarnoja18b_sac}.

\subsubsection{Lagrangian Relaxation}
The constrained problem \eqref{eq:RL_obj}--\eqref{eq:RL_const_so} is transformed into a min-max problem via Lagrangian duality 
\begin{align}\label{eq:lagrangian}
    \min_{\boldsymbol{\lambda} \succeq 0} \max_{\pi} \mathcal{L}(\pi, \boldsymbol{\lambda}),
\end{align}
where the Lagrangian is 
\begin{align}
    \mathcal{L}(\pi, \boldsymbol{\lambda}) = J_{\mathrm{R}}(\pi) - \sum_{i \in \{\mathrm{co}, \mathrm{so}\}} \lambda_{i}(J_{i}(\pi) - \epsilon_{i}),
\end{align}
and $\boldsymbol{\lambda} = [\lambda_{\mathrm{co}}, \lambda_{\mathrm{so}}]^{\mathrm{T}} \succeq 0$ is the vector of Lagrange multipliers for the average outage constraints.

\subsubsection{Maximum Entropy Framework}
Following the SAC framework~\cite{haarnoja18_sac}, the objective is augmented with entropy regularization to enhance exploration. Since the secrecy rate is non-concave in the beamforming vector, this regularization promotes diverse action selection and mitigates convergence to local optima. The stochastic policy $\pi_{\boldsymbol{\theta}}$ is parameterized by the actor neural network weights $\boldsymbol{\theta}$; hence, the maximization with respect to $\pi$ in \eqref{eq:lagrangian} is performed by optimizing the actor-network parameters $\boldsymbol{\theta}$. The entropy-regularized objective becomes
\begin{align}\label{eq:entropy_obj}
    J_{\mathrm{ent}}(\boldsymbol{\theta}) = \mathbb{E}_{\pi_{\boldsymbol{\theta}}} \left[ \sum_{n=1}^{N} \left( r_n - \frac{1}{N} \sum_{i \in \{\mathrm{co}, \mathrm{so}\}} \lambda_i c_{i,n} + \alpha \mathcal{H}(\pi_{\boldsymbol{\theta}}(\cdot|s_n)) \right) \right],
\end{align}
where $\alpha > 0$ is the temperature parameter scaling the entropy term, with larger $\alpha$ encouraging exploration and smaller $\alpha$ encouraging exploitation, and $\mathcal{H}(\pi_{\boldsymbol{\theta}}(\cdot|s_n)) = -\mathbb{E}_{\tilde{q}_n \sim \pi_{\boldsymbol{\theta}}(\cdot|s_n)}[\log \pi_{\boldsymbol{\theta}}(u_n, v_n\,|\,s_n)]$ is the entropy of the policy, which quantifies the randomness of the sampled actions. Throughout, $\mathbb{E}_{x \sim p}[\cdot]$ denotes the expectation with respect to $x$ sampled from the distribution $p$.

\begin{table}[t]
    \centering
    \renewcommand{\arraystretch}{1.3}
    \caption{Network Architectures. The reward and cost critics each comprise two networks with the listed architecture.}
    \label{tab:network_arch}
    \resizebox{\columnwidth}{!}{%
        \begin{tabular}{|c|c|c|c|}
            \hline
            \textbf{Layer} & \textbf{Actor}                                                     & \textbf{Reward Critic}                               & \textbf{Cost Critic}                                 \\
            \hline
            Input          & $s \in \mathbb{R}^{D_\mathrm{s}}$                                  & $(s, \mathbf{w}) \in \mathbb{R}^{D_\mathrm{s} + 2M}$ & $(s, \mathbf{w}) \in \mathbb{R}^{D_\mathrm{s} + 2M}$ \\
            Hidden 1       & Linear (256)+ReLU                                                  & Linear (256)+ReLU                                    & Linear (256)+ReLU                                    \\
            Hidden 2       & Linear (256)+ReLU                                                  & Linear (256)+ReLU                                    & Linear (256)+ReLU                                    \\
            Output         & $(\boldsymbol{\mu}, \boldsymbol{\sigma}) \in \mathbb{R}^{2(2M+1)}$ & $Q^{\mathrm{R}} \in \mathbb{R}$                      & $\mathbf{Q}^{\mathrm{C}} \in \mathbb{R}^{2}$         \\
            \hline
        \end{tabular}%
    }
\end{table}

\subsubsection{Network Architecture}
As summarized in Table~\ref{tab:network_arch}, the proposed SAC agent employs the following neural networks.
\begin{itemize}
    \item \textbf{Actor} $\pi_{\boldsymbol{\theta}}(\tilde{q}_n|s_n)$: A neural network that outputs the mean $\boldsymbol{\mu}_{\boldsymbol{\theta}}(s_n)$ and log standard deviation $\log \boldsymbol{\sigma}_{\boldsymbol{\theta}}(s_n)$ of a $(2M{+}1)$-dimensional diagonal Gaussian over the latent $\tilde{q}_n$, which is deterministically mapped to the squashed action $(u_n, v_n)$ via the tanh and sigmoid transformations in \eqref{eq:action_dir} and \eqref{eq:action_power}. Latent actions are sampled via the reparameterization trick \cite{haarnoja18_sac} as $\tilde{q}_n = \boldsymbol{\mu}_{\boldsymbol{\theta}}(s_n) + \boldsymbol{\sigma}_{\boldsymbol{\theta}}(s_n) \odot \boldsymbol{\xi}$, $\boldsymbol{\xi} \sim \mathcal{N}(\mathbf{0}, \mathbf{I})$, then mapped to $\mathbf{w}_n$ via \eqref{eq:action_dir}--\eqref{eq:action_norm}.
    \item \textbf{Reward Critics} $Q^{\mathrm{R}}_b(s_n,\mathbf{w}_n)$, $b \in \{1,2\}$: Two networks, parameterized by $\boldsymbol{\theta}^{\mathrm{R}}_b$, whose minimum is taken as a pessimistic estimate that mitigates the overestimation bias of the value function~\cite{fujimoto18_td3} for the policy update.
    \item \textbf{Cost Critics} $\mathbf{Q}^{\mathrm{C}}_b(s_n,\mathbf{w}_n) \in \mathbb{R}^2$, $b \in \{1,2\}$: Two networks, parameterized by $\boldsymbol{\theta}^{\mathrm{C}}_b$, each with two output heads; the components $[\mathbf{Q}^{\mathrm{C}}_b]_{\mathrm{co}}$ and $[\mathbf{Q}^{\mathrm{C}}_b]_{\mathrm{so}}$ estimate the connection and secrecy outage costs, respectively. With $\gamma_{\mathrm{d}} = 1$, $[\mathbf{Q}^{\mathrm{C}}_b]_i = \frac{1}{N}\sum_{n'=n}^{N}\mathbb{E}_{\pi_{\boldsymbol{\theta}}}[c_{i,n'} \mid s_n, \mathbf{w}_n]$, $i\in\{\mathrm{co},\mathrm{so}\}$, which matches the scale of the average constraints $J_i$ in \eqref{eq:RL_const_co} and \eqref{eq:RL_const_so}. The element-wise maximum is then taken as a pessimistic cost estimate in the policy update \eqref{eq:actor_loss}. In contrast to the minimum used for the reward critics, this maximum prevents underestimation of the costs~\cite{yang21_wcsac}.
\end{itemize}

\subsubsection{Training Updates}
The complete training procedure is summarized in Algorithm~\ref{alg:pd_sac}, where the specific network updates are performed as follows:
\begin{itemize}
    \item \textbf{Critic Update:} The reward critics are trained by minimizing the mean-squared error between their output and the regression target:
          \begin{align}\label{eq:critic_loss}
              L(\boldsymbol{\theta}^{\mathrm{R}}_b) = \mathbb{E}_{(s_n,\mathbf{w}_n,r_n,s_{n+1}) \sim \mathcal{R}} \left[ \left( Q^{\mathrm{R}}_b(s_n,\mathbf{w}_n) - y_{\mathrm{R}} \right)^2 \right],
          \end{align}
          where the expectation is taken over transition samples drawn from the replay buffer $\mathcal{R}$, and the target is
          \begin{align}
              &y_{\mathrm{R}} = r_n + (1\!-\!d) \,\nonumber\\
              &\times \mathbb{E}_{\tilde{q}_{n+1} \sim \pi_{\boldsymbol{\theta}}(\cdot\,|\,s_{n+1})} \!\bigl[ \min_{b\in\{1,2\}} \bar{Q}^{\mathrm{R}}_b(s_{n+1},f(\tilde{q}_{n+1})) \nonumber\\
              &\qquad\quad -\, \alpha \log \pi_{\boldsymbol{\theta}}(u_{n+1}, v_{n+1}\,|\,s_{n+1}) \bigr].
          \end{align}
          Here, $f(\tilde{q}) = \mathbf{w}$ denotes the deterministic mapping from the latent action to the transmit beamforming vector via \eqref{eq:action_dir}--\eqref{eq:action_norm}, $d \in \{0,1\}$ is the episode-termination indicator, and $\bar{\boldsymbol{\theta}}^{\mathrm{R}}_b$ denotes the target parameters maintained as an exponential moving average (EMA) of the online parameters $\boldsymbol{\theta}^{\mathrm{R}}_b$. The cost critic is updated similarly by minimizing the mean-squared error between its output and the regression target as
          \begin{align}\label{eq:cost_target}
              \mathbf{y}_{\mathrm{C}} = \frac{1}{N}\mathbf{c}_n + (1 - d) \, \max_{b\in\{1,2\}} \bar{\mathbf{Q}}^{\mathrm{C}}_b(s_{n+1}, f(\tilde{q}_{n+1})),
          \end{align}
          where $\bar{\boldsymbol{\theta}}^{\mathrm{C}}_b$ denotes the corresponding target parameters for the cost critic, and the element-wise maximum over the two cost critics acts as a pessimistic target to avoid underestimating the costs.

          \begin{algorithm}[t]
    \caption{PD-SAC for Secure Beamforming}
    \label{alg:pd_sac}
    \begin{algorithmic}[1]
        \REQUIRE Thresholds $\epsilon_{\mathrm{co}}$, $\epsilon_{\mathrm{so}}$; learning rates $\eta_{\theta}$, $\eta_{\theta^{\mathrm{R}}}$, $\eta_{\theta^{\mathrm{C}}}$, $\eta_{\lambda}$, $\eta_{\alpha}$; transmit power $P_{\mathrm{max}}$; target entropy $\bar{\mathcal{H}} = -(2M+1)$; soft update rate $\rho$; cost EMA decay $\chi$
        \STATE Initialize actor $\pi_{\boldsymbol{\theta}}$, reward critics $Q^{\mathrm{R}}_{1,2}$ with parameters $\boldsymbol{\theta}^{\mathrm{R}}_{1,2}$, cost critics $\mathbf{Q}^{\mathrm{C}}_{1,2}$ with parameters $\boldsymbol{\theta}^{\mathrm{C}}_{1,2}$ and two output heads each, and temperature $\alpha$
        \STATE Initialize target networks $\bar{\boldsymbol{\theta}}^{\mathrm{R}}_{1,2} \leftarrow \boldsymbol{\theta}^{\mathrm{R}}_{1,2}$, $\bar{\boldsymbol{\theta}}^{\mathrm{C}}_{1,2} \leftarrow \boldsymbol{\theta}^{\mathrm{C}}_{1,2}$
        \STATE Initialize dual variables $\tilde{\lambda}_{\mathrm{co}} \leftarrow \tilde{\lambda}_0$, $\tilde{\lambda}_{\mathrm{so}} \leftarrow \tilde{\lambda}_0$; cost EMA $\hat{c}_i \leftarrow \epsilon_i$; replay buffer $\mathcal{R} \leftarrow \emptyset$
        \FOR{each episode}
        \STATE Reset environment; observe initial state $s_1$
        \FOR{$n = 1$ to $N$}
        \STATE Sample latent action $\tilde{q}_n \sim \pi_{\boldsymbol{\theta}}(\cdot|s_n)$; compute $\mathbf{w}_n$ via \eqref{eq:action_dir}--\eqref{eq:action_norm}
        \STATE Execute $\mathbf{w}_n$; observe $r_n$, $c_{\mathrm{co},n}$, $c_{\mathrm{so},n}$, $s_{n+1}$
        \STATE Update cost EMA: $\hat{c}_i \leftarrow \chi\, \hat{c}_i + (1-\chi)\, \bar{c}_{i,n}$ \hfill $\triangleright$ Eq.~\eqref{eq:cost_ema}
        \STATE Store $(s_n, \mathbf{w}_n, r_n, c_{\mathrm{co},n}, c_{\mathrm{so},n}, s_{n+1}, d)$ in $\mathcal{R}$
        \IF{update step}
        \STATE Sample minibatch from $\mathcal{R}$
        \STATE Update reward critics via \eqref{eq:critic_loss} and cost critics via the MSE to \eqref{eq:cost_target}
        \STATE Update actor by minimizing \eqref{eq:actor_loss}
        \STATE Update dual: minimize $L_{\mathrm{dual}}(\tilde{\boldsymbol{\lambda}})$ in \eqref{eq:dual_update_co} via gradient descent
        \STATE Update temperature $\alpha$ by minimizing \eqref{eq:temp_loss}
        \STATE Soft update targets: $\bar{\boldsymbol{\theta}}^{\mathrm{R}}_b \leftarrow \rho \boldsymbol{\theta}^{\mathrm{R}}_b + (1-\rho)\bar{\boldsymbol{\theta}}^{\mathrm{R}}_b$, $\bar{\boldsymbol{\theta}}^{\mathrm{C}}_b \leftarrow \rho \boldsymbol{\theta}^{\mathrm{C}}_b + (1-\rho)\bar{\boldsymbol{\theta}}^{\mathrm{C}}_b$, $b\in\{1,2\}$
        \ENDIF
        \ENDFOR
        \ENDFOR
    \end{algorithmic}
\end{algorithm}

    \item \textbf{Actor Update:} The policy is updated by minimizing the following loss, which corresponds to maximizing the entropy-regularized objective $J_{\mathrm{ent}}$ in \eqref{eq:entropy_obj} with the cumulative reward and cost terms estimated by the critics
          \begin{align}\label{eq:actor_loss}
              &L(\boldsymbol{\theta}) = \mathbb{E}_{s_n \sim \mathcal{R},\, \tilde{q}_n \sim \pi_{\boldsymbol{\theta}}(\cdot\,|\,s_n)} \Big[ \alpha \log \pi_{\boldsymbol{\theta}}(u_n, v_n\,|\,s_n) \nonumber \\
              & - \min_{b\in\{1,2\}} Q^{\mathrm{R}}_b(s_n, f(\tilde{q}_n)) + \boldsymbol{\lambda}^{\mathrm{T}} \max_{b\in\{1,2\}} \mathbf{Q}^{\mathrm{C}}_b(s_n, f(\tilde{q}_n)) \Big],
          \end{align}
          where the element-wise maximum over the two cost critics matches the pessimistic cost target in \eqref{eq:cost_target}.
          Following the standard SAC formulation \cite{haarnoja18_sac}, the action log-density is
          \begin{align}\label{eq:log_prob}
              &\log  \pi_{\boldsymbol{\theta}}(u_n, v_n\,|\,s_n) \!=\! \sum_{i=1}^{2M} \Bigl[\log p(\tilde{q}_{n,i}^{\mathrm{dir}})\!  -\! \log\bigl(1\!-\!\tanh^2(\tilde{q}_{n,i}^{\mathrm{dir}})\bigr)\Bigr] \nonumber\\ &+ \log p(\tilde{q}_n^{\mathrm{pow}})  - \log \Bigl[P_{\mathrm{max}}\, \sigma(\tilde{q}_n^{\mathrm{pow}})\bigl(1-\sigma(\tilde{q}_n^{\mathrm{pow}})\bigr)\Bigr],
          \end{align}
          where $p(\cdot)$ is the latent Gaussian density of the pre-activation samples, and the remaining terms are the log-Jacobian correction terms induced by the tanh and sigmoid transformations, which convert the latent density into the density of the transformed action $(u_n,v_n)$. The $\ell_2$-normalization $\hat{\mathbf{d}}_n = u_n/\|u_n\|$ in \eqref{eq:action_dir} that produces the unit-norm direction is a deterministic post-processing step; its gradient with respect to $\boldsymbol\theta$ propagates through the actor loss \eqref{eq:actor_loss} via the chain rule, as $f$ enters both the reward and cost critic terms.\footnote{The reference implementation pre-scales $u$ by $\sqrt{P_{\mathrm{max}}}$ before $\ell_2$-normalization for numerical conditioning; since the normalization erases any positive scalar, $\mathbf{w}_n$ is unchanged and the scaling contributes only a state-independent additive constant to $\log\pi_{\boldsymbol\theta}$, which has no effect on the actor gradient.} The temperature $\alpha$ is adjusted via \eqref{eq:temp_loss}.

    \item \textbf{Dual Update:} The Lagrange multipliers are updated via gradient descent on the dual loss. To ensure non-negativity, i.e., $\lambda_i \geq 0$, we adopt a log-space parameterization $\lambda_i = \exp(\tilde{\lambda}_i)$ with $\tilde{\lambda}_i \in \mathbb{R}$, stacked as $\tilde{\boldsymbol{\lambda}} = [\tilde{\lambda}_{\mathrm{co}}, \tilde{\lambda}_{\mathrm{so}}]^{\mathrm{T}}$. The log-space parameters are initialized to a small value $\tilde{\lambda}_0 < 0$, so that the multipliers begin near zero and the policy first learns basic beamforming before constraint pressure is gradually applied.

          To reduce the bias from buffer samples drawn under earlier policies, we maintain an on-policy EMA of the per-step costs as
          \begin{align}\label{eq:cost_ema}
              \hat{c}_i \leftarrow \chi\, \hat{c}_i + (1 - \chi)\, \bar{c}_{i,n}, \quad i \in \{\mathrm{co}, \mathrm{so}\}
          \end{align}
          where $\bar{c}_{i,n}$ is the average cost across parallel environments at step $n$, and $\chi \in (0,1)$ is the decay factor. The dual loss is then
          \begin{align}
              L_{\mathrm{dual}}(\tilde{\boldsymbol{\lambda}}) = -\sum_{i \in \{\mathrm{co}, \mathrm{so}\}} \lambda_i \left(\hat{c}_i - \epsilon_i\right), \label{eq:dual_update_co}
          \end{align}
          which is minimized with respect to $\tilde{\boldsymbol{\lambda}}$ via gradient descent with learning rate $\eta_{\lambda}$.

    \item \textbf{Temperature Update:} The temperature $\alpha$ is adjusted to maintain a target entropy $\bar{\mathcal{H}} = -\dim(\tilde{q}_n) = -(2M+1)$, as is conventional for continuous action spaces~\cite{haarnoja18b_sac}, by minimizing
          \begin{align} \label{eq:temp_loss}
               L(\alpha) = \mathbb{E}_{s_n \sim \mathcal{R},\,\tilde{q}_n \sim \pi_{\boldsymbol{\theta}}(\cdot\,|\,s_n)} \Big[ &-\alpha \big( \log \pi_{\boldsymbol{\theta}}(u_n, v_n\,|\,s_n) + \bar{\mathcal{H}} \big) \Big].
          \end{align}
          When the policy entropy falls below $\bar{\mathcal{H}}$, $\alpha$ increases to encourage more stochastic actions, and vice versa.
\end{itemize}

\begin{rem}[Dual Update Dynamics and Stability]\label{rem:dual_dynamics}
              Since $\partial L_{\mathrm{dual}}/\partial \tilde{\lambda}_i = -\lambda_i (\hat{c}_i - \epsilon_i)$, the log-space parameter $\tilde{\lambda}_i$ increases when the EMA cost exceeds the threshold $\epsilon_i$, raising $\lambda_i$ and strengthening the penalty on constraint-violating actions. Conversely, when the EMA cost $\hat{c}_i$ is within budget, $\lambda_i$ decreases, allowing the policy to prioritize reward maximization. The on-policy EMA $\hat{c}_i$ in \eqref{eq:dual_update_co} mitigates oscillatory dual updates caused by stale replay-buffer samples and stabilizes the primal-dual training process.
          \end{rem}


\subsection{Computational Complexity Analysis}\label{sec:complexity}

Let $D_{\tilde{q}} = 2M+1$ denote the latent action dimension, and $H$ the hidden layer size.
For the proposed algorithm, during inference, the actor network requires a single forward pass with complexity $\mathcal{O}(D_\mathrm{s} H + H^2 + H D_{\tilde{q}})$. Since $D_{\tilde{q}} \ll H$ in practice, this simplifies to $\mathcal{O}(D_\mathrm{s} H + H^2)$. During training, the dominant cost per update step arises from the forward and backward passes through the actor and critic networks, yielding a complexity of $\mathcal{O}(|\mathcal{B}|((D_\mathrm{s} + 2M) H + H^2))$, where $|\mathcal{B}|$ is the minibatch size. The dual variable updates use only the scalar EMA costs $\hat{c}_i$ and require $\mathcal{O}(1)$ operations. Thus, the overall per-step training complexity is $\mathcal{O}(|\mathcal{B}|((D_\mathrm{s} + 2M)H + H^2))$.

For MRT, the beamforming vector is computed as $\mathbf{w}_n = \sqrt{P_{\mathrm{max}}} \mathbf{a}_{\ser,n} / \|\mathbf{a}_{\ser,n}\|$, which requires $\mathcal{O}(M)$ operations per time slot. For ZF, the beamforming vector is obtained by projecting the serving array response onto the null space of the eavesdropper array response matrix, with complexity $\mathcal{O}(M^2 E + M E^2 + E^3)$ per slot. For the per-slot SCA benchmark, each transmission slot is solved by a multistart SCA with $R_{\mathrm{SCA}}$ initializations, each refined over $K_{\mathrm{SCA}}$ outer iterations; every outer iteration solves a convex second-order cone program of dimension $2M$ by an interior-point method whose $K_{\mathrm{IP}}$ iterations are each dominated by an $\mathcal{O}(M^3)$ KKT factorization, yielding a per-slot cost of $\mathcal{O}(K_{\mathrm{SCA}} R_{\mathrm{SCA}} K_{\mathrm{IP}} M^3)$.

While the proposed RL approach incurs an offline training stage, its online execution complexity is strictly $\mathcal{O}(D_\mathrm{s} H + H^2)$, scaling linearly with the number of antennas $M$. By contrast, ZF requires algebraic projections that scale as $\mathcal{O}(M^2 E + M E^2 + E^3)$ and becomes infeasible as $E$ approaches $M$, while per-slot SCA scales as $\mathcal{O}(K_{\mathrm{SCA}} R_{\mathrm{SCA}} K_{\mathrm{IP}} M^3)$ and is unsuitable for real-time execution. The trained RL policy outputs the beamformer via a single forward pass, well-suited for real-time deployment under high satellite mobility.

 \begin{table}[t]
    \centering
    \renewcommand{\arraystretch}{1.2}
    \caption{Simulation Parameters}
    \label{Table:Sim_Param}
    \begin{tabular}{|l|c|}
        \hline
        \textbf{Parameter}                                         & \textbf{Value}                                                  \\
        \hline
        Earth's rotation rate \(\omega_{\mathrm{E}}\)
                                                                   & \(7.2921150 \times 10^{-5}\) rad/s                              \\
        \hline
        Radius of Earth $\re$                                      & 6,378 km                                                        \\
        \hline
        Gravitational constant $G$                                 & $6.674 \times 10^{-11}$ $\mathrm{m}^3/\mathrm{kg}/\mathrm{s}^2$ \\
        \hline
        Mass of Earth $M_{\mathrm{E}}$                             & $5.972 \times 10^{24}$ kg                                       \\
        \hline
        Speed of light $c$                                         & $3 \times 10^8$ m/s                                             \\
        \hline
        Noise spectral density $N_0$                               & $-174$ dBm/Hz                                                   \\
        \hline
        Altitude of the serving satellite $\as$                    & 600 km                                                          \\
        \hline
        Carrier frequency $\fc$                                    & 2 GHz                                                           \\
        \hline
        Path-loss exponent $\kappa$                                & 2                                                               \\
        \hline
        Maximum receive antenna gain $\Gkmax$                      & 24 dBi                                                          \\
        \hline
        Maximum transmit power $P_{\mathrm{max}}$                  & 40 dBm (10 W)                                                   \\
        \hline
        Avg.\ connection outage threshold $\epsilon_{\mathrm{co}}$ & 0.3                                                             \\
        \hline
        Avg.\ secrecy outage threshold $\epsilon_{\mathrm{so}}$    & 0.3                                                             \\
        \hline
        Bandwidth $W$                                              & 100 MHz                                                         \\
        \hline
        Time slot duration $\delta$                                & 1 s                                                             \\
        \hline
        Evaluated eavesdropper counts $E$                          & $\{1,2,3,4,5,6,7\}$                                             \\
        \hline
        Number of antennas $M = M_{\bar{x}} \times M_{\bar{y}}$    & $4 \times 4 = 16$                                               \\
        \hline
        Nakagami fading parameter $m_{\ser}, m_{\eve_j}$           & 2                                                               \\
        \hline
        Target service rate $R_{\ser,n}$                           & 0.5 bps/Hz                                                      \\
        \hline
        Target eavesdropper rate $R_{\eve_j,n}$                    & 1.0 bps/Hz                                                      \\
        \hline
        Serving satellite inclination $i_{\ser}$                   & $89.5^{\circ}$                                                  \\
        \hline
        Serving satellite RAAN $\Omega_{\ser}$                     & $45^{\circ}$                                                    \\
        \hline
        Eavesdropper altitude $a_{\eve_j}$                        & 600 km                                                          \\
        \hline
        Eavesdropper inclination $i_{\eve_j}$                     & $89^{\circ}$                                                    \\
        \hline
        Eavesdropper RAAN $\Omega_{\eve_j}$                       & $90^{\circ}$                                                    \\
        \hline
        Initial position offset $\Delta u_{\eve_j}$               & $|\Delta u_{\eve_j}| \leq 5^{\circ}$                         \\
        \hline
        Terminal latitude $\phi$                                   & $90^{\circ}$ (North Pole)                                       \\
        \hline
        3-dB beamwidth angle $\nu^{\mathrm{3dB}}$                  & $15^{\circ}$                                                    \\
        \hline
        Discount factor $\gamma_{\mathrm{d}}$                      & 1.0                                                             \\
        \hline
        Soft update rate $\rho$                                    & 0.005                                                           \\
        \hline
        Learning rate (Actor/Critic) $\eta_{\theta}, \eta_{\theta^{\mathrm{R}}}$  & $3 \times 10^{-4}$                                              \\
        \hline
        Lagrange-multiplier learning rate $\eta_{\lambda}$            & $3 \times 10^{-3}$                                              \\
        \hline
        Replay buffer size                                         & $10^6$                                                          \\
        \hline
        Batch size                                                 & 256                                                             \\
        \hline
        Target entropy $\bar{\mathcal{H}}$                         & $-(2M+1)$                                                       \\
        \hline
        Number of parallel environments                            & 50                                                              \\
        \hline
        Hidden layer size                                          & 256                                                             \\
        \hline
        Number of hidden layers                                    & 2                                                               \\
        \hline
        Lagrange-multiplier warmup steps $T_{\mathrm{warm}}$                      & $2 \times 10^{4}$                                               \\
        \hline
        Initial log-space Lagrange multiplier $\tilde{\lambda}_0$                              & $-3.0$                                                          \\
        \hline
        $\lambda$ clamp $[\lambda_{\min}, \lambda_{\max}]$         & $[0.01, 100]$                                                   \\
        \hline
        Cost EMA decay $\chi$                                      & 0.995                                                           \\
        \hline
        Actor EMA rate $\xi$                                       & 0.005                                                           \\
        \hline
    \end{tabular}
\end{table}

\section{Simulation Results}\label{sec:sim_res}

The simulation parameters listed in Table~\ref{Table:Sim_Param} are used unless otherwise stated.  
A separate PD-SAC policy is trained for each~$E$ on the deployment scenario using 50 parallel environment instances.\footnote{While the instantaneous CSI of eavesdroppers is unavailable, their cardinality $E$ is observable from the tracked LEO constellation, since adversarial satellites can be enumerated via ephemeris-based space situational awareness even when their channels remain uncertain.}
We compare the proposed PD-SAC against MRT, ZF, per-slot SCA, and PD-PPO.  For the per-slot SCA benchmark, the non-convex problem at each transmission slot, which maximizes the secrecy rate evaluated with the geometry-derived average SNRs subject to the power constraint and the per-slot counterparts $c_{\mathrm{co},n} \leq \epsilon_{\mathrm{co}}$ and $c_{\mathrm{so},n} \leq \epsilon_{\mathrm{so}}$ of the average outage constraints, is solved by a multistart SCA, where the secrecy rate is iteratively linearized and each convex subproblem is solved by an interior-point method~\cite{markswright78}. The resulting beamformers are evaluated in the same simulation environment as the learned policies.  The on-policy alternative PD-PPO shares the same CMDP formulation, network architecture, and dual variable structure as PD-SAC~\cite{schulman2017proximal}, with PPO-specific hyperparameters tuned for stable convergence. 

Several measures are adopted to stabilize the primal-dual training. The Lagrange multipliers are clamped to $[\lambda_{\min}, \lambda_{\max}]$ to keep the constraint penalty active while avoiding oscillatory updates. The dual variables are further frozen during an initial warmup of $T_{\mathrm{warm}}$ training steps. For evaluation and deployment, the EMA-averaged policy $\pi_{\bar{\boldsymbol{\theta}}}$ is used in place of the training policy to smooth short-term oscillations induced by the primal-dual dynamics. The corresponding hyperparameters are listed in Table~\ref{Table:Sim_Param}.

Fig.~\ref{fig:episode_RL} shows the training convergence of PD-SAC for $E=3$. In Fig.~\ref{fig:episode_RL}(a), PD-SAC surpasses the MRT baseline within the first few episodes and the ZF baseline after around $50$ episodes, converging to approximately $2.8$ \,bps/Hz, within roughly $7\%$ of the offline SCA benchmark. 
As shown in Figs.~\ref{fig:episode_RL}(b) and~\ref{fig:episode_RL}(c), PD-SAC rapidly drives the connection outage below the threshold, while the secrecy outage is steered toward $\epsilon_{\mathrm{so}}=0.3$ as the dual variable actively enforces the constraint boundary. The primal-dual updates remove the need for manual penalty tuning.

\begin{figure}[t]
    \centering
    \includegraphics[width=.950\columnwidth]{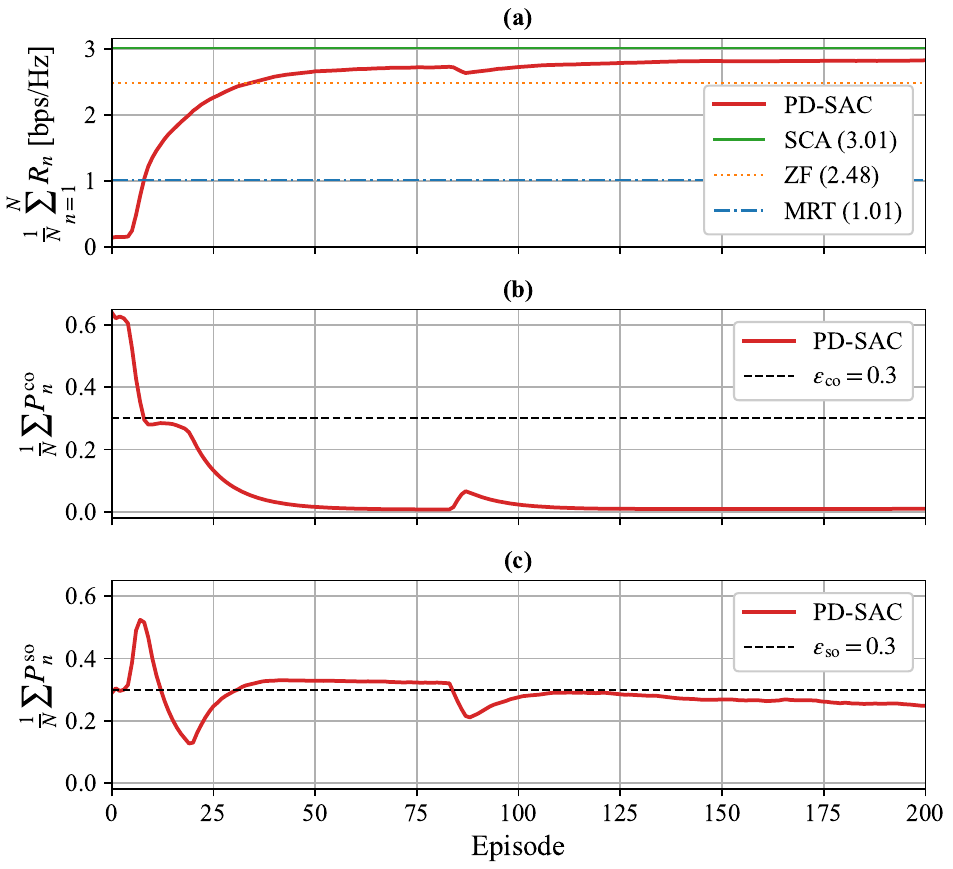}
    \caption{Training convergence of PD-SAC for the $E=3$: (a)~average secrecy rate, (b)~average connection outage probability, and (c)~average secrecy outage probability. The dashed lines in (b) and (c) indicate the average outage constraint thresholds $\epsilon_{\mathrm{co}} = \epsilon_{\mathrm{so}} = 0.3$.}
    \label{fig:episode_RL}
\end{figure}

\begin{figure}[t]
    \centering
    \includegraphics[width=.950\columnwidth]{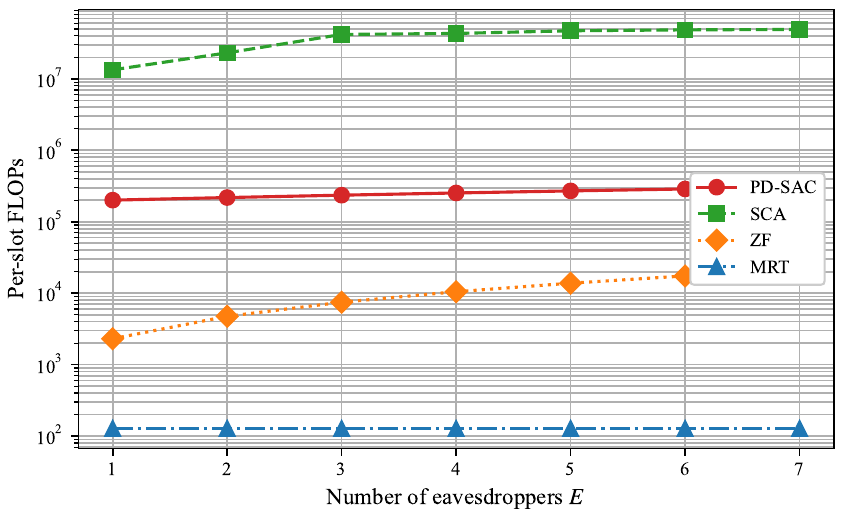}
    \caption{Per-slot FLOP counts versus the number of eavesdroppers $E$. }
    \label{fig:complexity}
\end{figure}

\begin{figure*}[t]
    \centering
    \includegraphics[width=.950\textwidth]{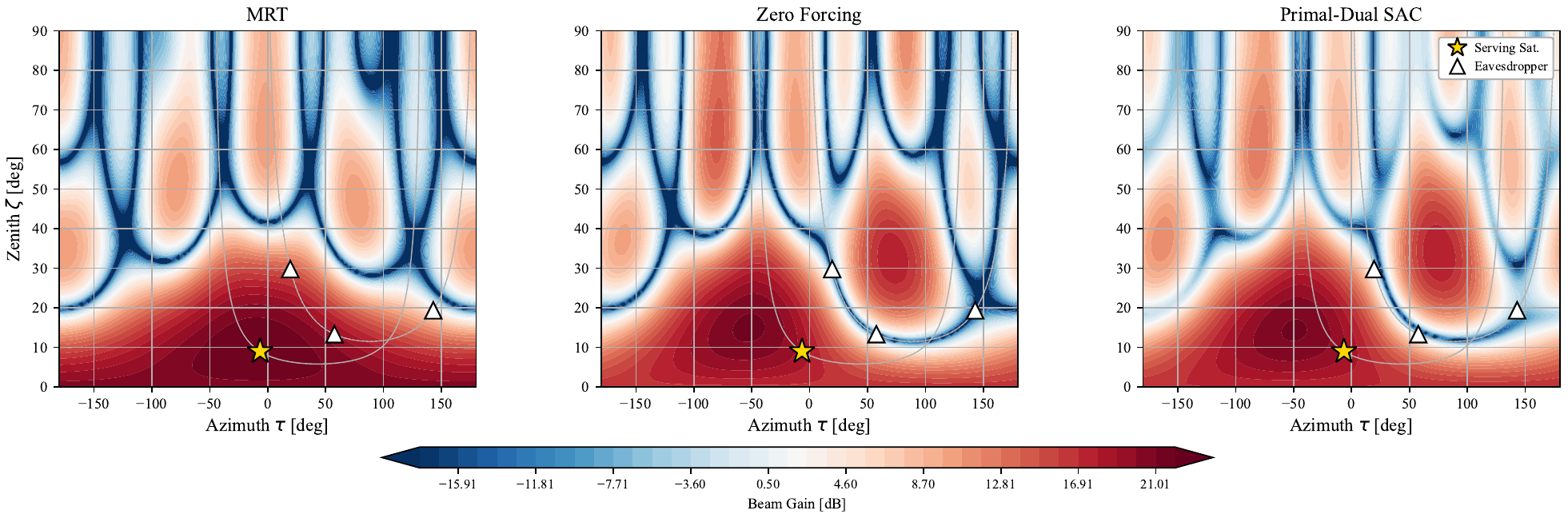}
    \caption{Beam patterns of MRT, ZF, and the proposed PD-SAC at time slot $n=376$. The red star and cyan triangles represent the directions of the serving satellite and the eavesdroppers, respectively. 
    The eavesdropper orbital parameters are set to $a_{\eve_j}=600$\,km, $i_{\eve_j}=89^{\circ}$, $\Omega_{\eve_j}=90^{\circ}$, and $\Delta u_{\eve_j}=\{+2^{\circ}, -2^{\circ}, 0^{\circ}\}$.}
    \label{Fig:beam_pattern}
\end{figure*}

Fig.~\ref{fig:complexity} shows the per-slot floating-point operation (FLOP) count for each method across $E \in \{1,\ldots,7\}$. The dominant FLOP counts are approximately given by $8M$ for MRT, $8M^2 E + 8M E^2 + 2 E^3 + 8M$ for ZF, $2[D_\mathrm{s} H + H^2 + H \cdot 2(2M+1)]$ for PD-SAC, and $R_{\mathrm{SCA}} K_{\mathrm{SCA}} K_{\mathrm{IP}} [8M^3 + \mathcal{O}(M^2 E)]$ for SCA, consistent with the complexity analysis in Section~\ref{sec:complexity}. Using $R_{\mathrm{SCA}}=10$ restarts per slot and the solver-measured averages $K_{\mathrm{SCA}}\approx 9.7$ and $K_{\mathrm{IP}}\approx 12$, at $E=3$ PD-SAC inference requires approximately $2.4\times 10^5$ FLOPs whereas the per-slot SCA solve requires about $4.2\times 10^7$ FLOPs, over two orders of magnitude larger.

      \begin{table}[t]
        \centering
        \caption{Performance comparison of beamforming schemes.}
        \label{tab:comparison}
        \renewcommand{\arraystretch}{1.2}
        \begin{tabular}{l|c|cc|cc}
            \hline
            \textbf{Method} & $\frac{1}{N}\sum_{n=1}^{N} R_n$ & \multicolumn{2}{c|}{$\frac{1}{N}\sum P_n^{\mathrm{co}}$} & \multicolumn{2}{c}{$\frac{1}{N}\sum P_n^{\mathrm{so}}$}                                   \\
                            &                                         & Bound                                                    & Exact                                                   & Bound          & Exact          \\
            \hline
            MRT             & 1.01                                    & \textbf{0.001}                                           & \textbf{0.000}                                          & 0.988          & 0.988          \\
            ZF              & 2.48                                    & 0.237                                                    & 0.199                                                   & \textbf{0.000} & \textbf{0.000} \\
            PD-PPO          & 2.18                                    & 0.025                                                    & 0.014                                                   & 0.232          & 0.226          \\
            SCA             & 3.01                                    & 0.015                                                    & 0.008                                                   & 0.112          & 0.090          \\
            \textbf{PD-SAC} & 2.80                                    & 0.011                                                    & 0.006                                                   & 0.232          & 0.222          \\
            \hline
        \end{tabular}
    \end{table}

    Table~\ref{tab:comparison} compares the average secrecy rate and outage probabilities for the five schemes with $E=3$. The ``Bound'' uses the upper bounds from Lemma~\ref{lem:gamma_approx}, and ``Exact'' uses the true incomplete gamma function. Among the deployable policies, PD-SAC attains an average secrecy rate of $2.8$ \,bps/Hz while satisfying both average outage constraints under the upper bound, namely $\frac{1}{N}\sum P_n^{\mathrm{co}}=0.011 \leq \epsilon_{\mathrm{co}}$ and $\frac{1}{N}\sum P_n^{\mathrm{so}}=0.232 \leq \epsilon_{\mathrm{so}}$. The corresponding exact values, $0.006$ and $0.222$, also lie below the thresholds. This corresponds to a $177\%$ gain over MRT and a $13\%$ gain over ZF.  SCA attains $3.01$\,bps/Hz as an offline benchmark; the proposed PD-SAC closes this gap to within $7\%$ through a single forward pass at inference, with much less complexity as shown in Fig.~\ref{fig:complexity}. PD-PPO attains $2.18$\,bps/Hz, $22\%$ below PD-SAC.
    MRT achieves the lowest connection outage since it maximizes the serving link gain, but it ignores the eavesdroppers entirely. Conversely, ZF completely eliminates the secrecy outage by nulling the eavesdropper channels, but this comes at the cost of a high connection outage, as the null constraint significantly reduces the beamforming gain toward the serving satellite.
    The evaluations using the derived bounds demonstrate close agreement with the exact formulations, confirming that the bound-based approximation incurs almost no practical performance loss.

        Fig.~\ref{Fig:beam_pattern} illustrates the beam gain $|\hkn^{\H} \mathbf{w}_{n}|^2$ as a function of the azimuth angle $\tau_{k,n}$ and the zenith angle $\zeta_{k,n}$ at time slot $n=376$. As expected, MRT consistently directs its main lobe toward the serving satellite but provides no suppression toward the eavesdroppers, resulting in high eavesdropper gain. ZF places nulls at the eavesdropper directions, but reduces the gain toward the serving satellite.
    The proposed PD-SAC steers the main lobe toward the serving satellite while partially nulling the eavesdropper directions to satisfy the average secrecy outage constraint.

    \begin{figure*}[t]
        \centering
        \includegraphics[width=.95\textwidth]{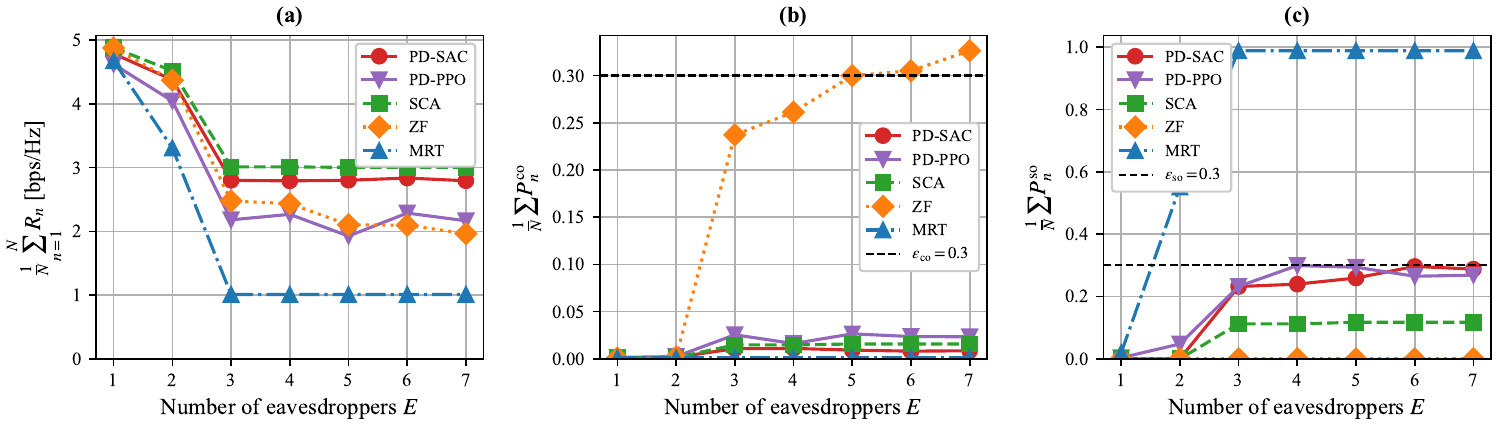}
        \caption{Performance versus the number of eavesdroppers $E \in \{1,\dots,7\}$: (a)~average secrecy rate, (b)~average connection outage probability, and (c)~average secrecy outage probability. Dashed lines indicate the constraint thresholds.}
        \label{fig:vs_E}
    \end{figure*}

    Fig.~\ref{fig:vs_E} shows the secrecy performance as the number of eavesdroppers $E$ increases. The proposed PD-SAC performs comparably to ZF for \(E=\{1,2\}\), but surpasses it with an increasing margin as \(E\) grows, and remains close to the offline SCA benchmark across the entire range of \(E\). In contrast, ZF's connection outage rises sharply as the null-space dimension shrinks with each additional eavesdropper, exceeding the $0.3$ threshold for $E \geq 6$, while MRT's secrecy outage rises steeply from $E \geq 2$ and saturates near one for $E \geq 3$, since it does not actively suppress eavesdroppers. PD-PPO achieves a lower secrecy rate than PD-SAC for \(E \geq 3\), with a gap that becomes pronounced as the number of eavesdroppers increases. Consequently, PD-SAC, PD-PPO, and SCA satisfy both outage constraints across the entire range, whereas ZF becomes infeasible at $E \geq 6$ due to connection outage, and MRT violates the secrecy outage constraint for $E \geq 2$.
    In addition, as the number of eavesdroppers increases, the secrecy rate eventually saturates because, in the non-colluding scenario, the secrecy performance is mainly determined by the most dominant eavesdropper rather than by all eavesdroppers collectively. In the considered along-track deployment, the large satellite-to-ground distances make the link geometry the dominant factor. Hence, adding more satellites beyond the strongest geometric eavesdropping positions causes only marginal additional degradation.

    \section{Conclusions}\label{sec:conclusions}
    This paper investigated secure uplink beamforming for LEO satellite networks against multiple satellite eavesdroppers. We derived exact outage probabilities under Nakagami-$m$ fading and developed tractable upper-bound cost functions to manage their intractability. By formulating the secrecy rate maximization as a CMDP, we proposed a PD-SAC algorithm to optimize the precoder. Simulations demonstrated that the proposed algorithm outperforms MRT and ZF, particularly as the number of eavesdroppers increases, while satisfying the outage constraints. It closely approaches an offline per-slot SCA benchmark while outperforming an on-policy PD-PPO counterpart. Although the SCA benchmark attains a higher secrecy rate, its iterative per-slot optimization is unsuitable for real-time deployment. In contrast, the proposed policy provides a deterministic mapping from the geometry-derived channel characteristics to the beamforming vector, enabling practical secure communications on dynamic satellite links. Future work includes addressing colluding eavesdroppers, imperfect channel state information, and massive antenna arrays.

\ifCLASSOPTIONcaptionsoff
    \newpage
\fi

\enlargethispage{2\baselineskip}
\bibliographystyle{IEEEtran}
\bibliography{
    references/IEEEabrv,
    references/3GPP,
    references/books,
    references/chSR,
    references/myPapers,
    references/refs,
    references/DRL,
    references/etc,
    references/antGain
}

@techreport{TR38.811,
  author      = {{3GPP TR 38.811 v15.4.0}},
  title       = {Study on {NR} to support non-terrestrial networks},
  institution = {3GPP},
  year        = {2020},
  month       = {Sep.}
}

@techreport{TR38.821,
  author      = {{3GPP TR 38.821 v16.0.0}},
  title       = {Solutions for {NR} to support non-terrestrial networks ({NTN})},
  institution = {3GPP},
  year        = {2019},
  month       = {Dec.}
}

@ARTICLE{23_DRL_5G,
  author={Faris B. Mismar and Brian L. Evans and Ahmed Alkhateeb},
  journal=TWC, 
  title={Deep Reinforcement Learning for {5G} Networks: Joint Beamforming, Power Control, and Interference Coordination}, 
  year={2020},
  volume={68},
  number={3},
  pages={1581-1592},
  doi={10.1109/TCOMM.2019.2961332},
  ISSN={0090-6778},
  month={Dec.},}

@ARTICLE{23_sat_rl_beam,
  author={Deng, Danhao and Wang, Chaowei and Pang, Mingliang and Wang, Weidong},
  journal=WCL, 
  title={Dynamic Resource Allocation With Deep Reinforcement Learning in Multibeam Satellite Communication}, 
  year={2023},
  volume={12},
  number={1},
  pages={75-79},
  doi={10.1109/LWC.2022.3217316}}

@ARTICLE{24_sat_rl_beam,
  author={Wu, Min and Guo, Kefeng and Li, Xingwang and Lin, Zhi and Wu, Yongpeng and Tsiftsis, Theodoros A. and Song, Houbing},
  journal=TCOM, 
  title={Deep Reinforcement Learning-Based Energy Efficiency Optimization for RIS-Aided Integrated Satellite-Aerial-Terrestrial Relay Networks}, 
  year={2024},
  volume={72},
  number={7},
  pages={4163-4178},
  doi={10.1109/TCOMM.2024.3370618}}

@inproceedings{haarnoja18_sac,
  author    = {T. Haarnoja and A. Zhou and P. Abbeel and S. Levine},
  title     = {Soft Actor-Critic: Off-Policy Maximum Entropy Deep Reinforcement Learning with a Stochastic Actor},
  booktitle = {Proc. 35th Int. Conf. Mach. Learn. (ICML)},
  volume    = {80},
  pages     = {1861--1870},
  year      = {2018},
  month     = {July},
  publisher = {PMLR}
}

@article{haarnoja18b_sac,
  author    = {T. Haarnoja and A. Zhou and K. Hartikainen and G. Tucker and S. Ha and J. Tan and V. Kumar and H. Zhu and A. Gupta and P. Abbeel and S. Levine},
  title     = {Soft Actor-Critic Algorithms and Applications},
  journal   = {arXiv preprint arXiv:1812.05905},
  year      = {2018},
  month     = {Dec.}
}

@article{schulman2017proximal,
  author    = {J. Schulman and F. Wolski and P. Dhariwal and A. Radford and O. Klimov},
  title     = {Proximal Policy Optimization Algorithms},
  journal   = {arXiv preprint arXiv:1707.06347},
  year      = {2017},
  month     = {July}
}

@inproceedings{achiam17_cpo,
  author    = {J. Achiam and D. Held and A. Tamar and P. Abbeel},
  title     = {Constrained Policy Optimization},
  booktitle = {Proc. 34th Int. Conf. Mach. Learn. (ICML)},
  volume    = {70},
  pages     = {22--31},
  year      = {2017},
  month     = {Aug.}
}

@article{paternain23_safe,
  author    = {S. Paternain and M. Calvo-Fullana and L. F. O. Chamon and A. Ribeiro},
  title     = {Safe Policies for Reinforcement Learning via Primal-Dual Methods},
  journal   = {{IEEE} Trans. Autom. Control},
  volume    = {68},
  number    = {3},
  pages     = {1321--1336},
  year      = {2023},
  month     = {Mar.},
  doi       = {10.1109/TAC.2022.3152724}
}

@inproceedings{yang21_wcsac,
  author    = {Q. Yang and T. D. Sim{\~a}o and S. H. Tindemans and M. T. J. Spaan},
  title     = {{WCSAC}: Worst-Case Soft Actor Critic for Safety-Constrained Reinforcement Learning},
  booktitle = {Proc. AAAI Conf. Artif. Intell. (AAAI)},
  volume    = {35},
  number    = {12},
  pages     = {10639--10646},
  year      = {2021},
  month     = {Feb.},
  publisher = {AAAI Press}
}

@INPROCEEDINGS{19_RCPO,
  author={Tessler, Chen and Mankowitz, Daniel J. and Mannor, Shie},
  booktitle={Proc. Int. Conf. Learn. Representations (ICLR)}, 
  title={Reward Constrained Policy Optimization}, 
  year={2019},
  month={May},
  address={New Orleans, LA, USA}
}

@INPROCEEDINGS{06_RL_constrained,
  author={Geibel, Peter},
  booktitle={Proc. Eur. Conf. Mach. Learn. (ECML)},
  title={Reinforcement Learning for {MDP}s with Constraints},
  year={2006},
  pages={646-653},
  month={Sep.},
  address={Berlin, Germany}
}

@inproceedings{fujimoto18_td3,
  author    = {S. Fujimoto and H. van Hoof and D. Meger},
  title     = {Addressing Function Approximation Error in Actor-Critic Methods},
  booktitle = {Proc. 35th Int. Conf. Mach. Learn. (ICML)},
  volume    = {80},
  pages     = {1587--1596},
  year      = {2018},
  month     = {Jul.},
  publisher = {PMLR}
}

@STRING{WCL         = "{IEEE} Wireless Commun. Lett."}

@STRING{JSAC        = "{IEEE} J. Sel. Areas Commun."}

@STRING{TCOM        = "{IEEE} Trans. Commun."}

@STRING{TWC         = "{IEEE} Trans. Wireless Commun."}

@STRING{TVT         = "{IEEE} Trans. Veh. Technol."}

@STRING{TIFS        = "{IEEE} Trans. Inf. Forensics Security"}

@STRING{CST         = "{IEEE} Commun. Surveys Tuts."}

@STRING{IA          = "{IEEE} Access"}

@STRING{ICML     = "Proc. Int. Conf. Mach. Learn. (ICML)"}

@STRING{ICLR     = "Proc. Int. Conf. Learn. Represent. (ICLR)"}

@STRING{AAAI     = "Proc. AAAI Conf. Artif. Intell. (AAAI)"}

@misc{antBessel,
      title={Max-Min Fair Energy-Efficient Beam Design for Quantized {ISAC} {LEO} Satellite Systems: A Rate-Splitting Approach}, 
      author={Ziang Liu and Longfei Yin and Wonjae Shin and Bruno Clerckx},
      year={2024},
      eprint={2402.09253},
      archivePrefix={arXiv},
      primaryClass={eess.SP},
      note={arXiv:2402.09253}
}

@book{book99_cmdp,
  author    = {E. Altman},
  title     = {Constrained {Markov} Decision Processes},
  publisher = {Chapman and Hall/CRC},
  address   = {Boca Raton, FL},
  year      = {1999}
}

@book{book13OrbitalVelocity,
  title={Fundamental planetary science: physics, chemistry and habitability},
  author={Lissauer, Jack J and De Pater, Imke},
  year={2013},
  publisher={Cambridge University Press}
}

@article{markswright78,
  author  = {B. R. Marks and G. P. Wright},
  title   = {A General Inner Approximation Algorithm for Nonconvex Mathematical Programs},
  journal = {Oper. Res.},
  volume  = {26},
  number  = {4},
  pages   = {681--683},
  year    = {1978}
}

@article{wyner75,
  author    = {A. D. Wyner},
  title     = {The Wire-Tap Channel},
  journal   = {Bell Syst. Tech. J.},
  volume    = {54},
  number    = {8},
  pages     = {1355--1387},
  year      = {1975},
  month     = {Oct.},
  doi       = {10.1002/j.1538-7305.1975.tb02040.x}
}

@article{mukherjee14,
  author    = {A. Mukherjee and S. A. A. Fakoorian and J. Huang and A. L. Swindlehurst},
  title     = {Principles of Physical Layer Security in Multiuser Wireless Networks: A Survey},
  journal   = CST,
  volume    = {16},
  number    = {3},
  pages     = {1550--1573},
  year      = {2014},
  doi       = {10.1109/COMST.2014.2303214}
}

@article{lin19_robust_sat,
  author    = {Z. Lin and M. Lin and J. Ouyang and W.-P. Zhu and A. D. Panagopoulos and M.-S. Alouini},
  title     = {Robust Secure Beamforming for Multibeam Satellite Communication Systems},
  journal   = TVT,
  volume    = {68},
  number    = {6},
  pages     = {6202--6206},
  year      = {2019},
  month     = {June},
  doi       = {10.1109/TVT.2019.2913793}
}

@article{guo20_pls_sat,
  author    = {K. Guo and K. An and B. Zhang and Y. Huang and X. Tang and G. Zheng and T. A. Tsiftsis},
  title     = {Physical Layer Security for Multiuser Satellite Communication Systems With Threshold-Based Scheduling Scheme},
  journal   = TVT,
  volume    = {69},
  number    = {5},
  pages     = {5129--5141},
  year      = {2020},
  month     = {May},
  doi       = {10.1109/TVT.2020.2979408}
}

@article{lin18_cog_sat,
  author    = {M. Lin and Z. Lin and W.-P. Zhu and J.-B. Wang},
  title     = {Joint Beamforming for Secure Communication in Cognitive Satellite Terrestrial Networks},
  journal   = JSAC,
  volume    = {36},
  number    = {5},
  pages     = {1017--1029},
  year      = {2018},
  month     = {May},
  doi       = {10.1109/JSAC.2018.2832819}
}

@ARTICLE{JSAC_Zhu,
  author={Zhu, Yongxu and Zheng, Gan and Fitch, Michael},
  journal=JSAC, 
  title={Secrecy Rate Analysis of {UAV}-Enabled mmWave Networks Using {Mat\'ern} Hardcore Point Processes}, 
  year={2018},
  volume={36},
  number={7},
  pages={1397-1409},
  doi={10.1109/JSAC.2018.2825158}}

@ARTICLE{TIFS_Lei,
  author={Lei, Jiang and Han, Zhu and Vazquez-Castro, Mar\'ia \'Angeles and Hjorungnes, Are},
  journal=TIFS, 
  title={Secure Satellite Communication Systems Design With Individual Secrecy Rate Constraints}, 
  year={2011},
  volume={6},
  number={3},
  pages={661-671},
  doi={10.1109/TIFS.2011.2148716}}

@ARTICLE{TWC_Zheng,
  author={Zheng, Gan and Arapoglou, Pantelis-Daniel and Ottersten, Bjorn},
  journal=TWC, 
  title={Physical Layer Security in Multibeam Satellite Systems}, 
  year={2012},
  volume={11},
  number={2},
  pages={852-863},
  doi={10.1109/TWC.2011.120911.111460}}

@ARTICLE{WCL_Li,
  author={Li, Bin and Fei, Zesong and Xu, Xiaoming and Chu, Zheng},
  journal=WCL, 
  title={Resource Allocations for Secure Cognitive Satellite-Terrestrial Networks}, 
  year={2018},
  volume={7},
  number={1},
  pages={78-81},
  doi={10.1109/LWC.2017.2755014}}

@article{tang19_uav_eve,
  author    = {J. Tang and G. Chen and J. P. Coon},
  title     = {Secrecy Performance Analysis of Wireless Communications in the Presence of {UAV} Jammer and Randomly Located {UAV} Eavesdroppers},
  journal   = TIFS,
  volume    = {14},
  number    = {11},
  pages     = {3026--3041},
  year      = {2019},
  month     = {Nov.},
  doi       = {10.1109/TIFS.2019.2912074}
}

@article{yuan19_uav_eve,
  author    = {X. Yuan and Z. Feng and W. Ni and Z. Wei and R. P. Liu and J. A. Zhang},
  title     = {Secrecy Rate Analysis Against Aerial Eavesdropper},
  journal   = TCOM,
  volume    = {67},
  number    = {10},
  pages     = {7027--7042},
  year      = {2019},
  month     = {Oct.},
  doi       = {10.1109/TCOMM.2019.2928558}
}

@article{bao20_uav_eve,
  author    = {T. Bao and J. Zhu and H.-C. Yang and M. O. Hasna},
  title     = {Secrecy Outage Performance of Ground-to-Air Communications With Multiple Aerial Eavesdroppers and Its Deep Learning Evaluation},
  journal   = WCL,
  volume    = {9},
  number    = {9},
  pages     = {1351--1355},
  year      = {2020},
  month     = {Sep.},
  doi       = {10.1109/LWC.2020.2993189}
}

@article{my24JSAC_EA,
  author  = {J. Seong and J. Park and D.-H. Jung and J. Park and W. Shin},
  title   = {Rate-splitting for joint unicast and multicast transmission in {LEO} satellite networks with non-uniform traffic demand},
  journal = JSAC,
  year    = {2024},
  note    = {early access, Sep. 13, 2024, doi: 10.1109/JSAC.2024.3460073}
}

@article{my24TWC,
  author   = {Jung, Dong-Hyun and Nam, Hongjae and Choi, Junil and Love, David J.},
  journal  = TWC,
  title    = {Modeling and Analysis of {GEO} Satellite Networks},
  year     = {2024},
  volume   = {23},
  number   = {11},
  pages    = {16757-16770},
  doi      = {10.1109/TWC.2024.3447229}
}

@article{my23WCL,
  author   = {Jung, Dong-Hyun and Ryu, Joon-Gyu and Choi, Junil},
  journal  = WCL,
  title    = {Satellite Clustering for Non-Terrestrial Networks: Orbital Configuration-Dependent Outage Analysis},
  year     = {2024},
  volume   = {13},
  number   = {2},
  pages    = {550-554},
  doi      = {10.1109/LWC.2023.3335918}
}

@article{my22TIFS,
  author  = {Jung, Dong-Hyun and Ryu, Joon-Gyu and Choi, Junil},
  journal = TIFS,
  title   = {When Satellites Work as Eavesdroppers},
  year    = {2022},
  volume  = {17},
  number  = {},
  pages   = {2784-2799},
  doi     = {10.1109/TIFS.2022.3188150},
  issn    = {1556-6021},
  month   = {}
}

@inproceedings{my17VTC,
  author    = {Jung, Donghyun and Lee, Jae Hong},
  booktitle = {2017 IEEE 86th Vehicular Technology Conference (VTC-Fall)},
  title     = {Secrecy Performance of Full-Duplex Relay System with Randomly Located Eavesdroppers},
  year      = {2017},
  volume    = {},
  number    = {},
  pages     = {1-5},
  doi       = {10.1109/VTCFall.2017.8288219}
}

@article{Alzer,
  title={On some inequalities for the incomplete gamma function},
  author={Alzer, Horst},
  journal={Mathematics of Computation},
  volume={66},
  number={219},
  pages={1239--1252},
  year={1997},
  publisher={American Mathematical Society}
}

\end{document}